\documentclass{llncs}
\usepackage{xcolor} 
\usepackage{tikz-qtree}
\usepackage{etex}
\usepackage{amsmath}
\usepackage{amssymb}
\usepackage{proof}
\usepackage{hyperref}
\usepackage{caption}
\usepackage{comment}
\usepackage{subcaption}
\usepackage{fancybox}
\usepackage{stmaryrd}

\usepackage{todonotes}
\newcommand{\knote}[1]{\todo[inline, color=blue!20]{#1}}

\definecolor{light-gray}{gray}{0.86}

\newcommand{\mylet}[2]{\mathbf{let}\ #1 \ \mathbf{in} \ #2}

\newcommand{\remph}[1]{\textcolor{blue}{#1}}

\makeatletter
\@ifundefined{lhs2tex.lhs2tex.sty.read}%
  {\@namedef{lhs2tex.lhs2tex.sty.read}{}%
   \newcommand\SkipToFmtEnd{}%
   \newcommand\EndFmtInput{}%
   \long\def\SkipToFmtEnd#1\EndFmtInput{}%
  }\SkipToFmtEnd

\newcommand\ReadOnlyOnce[1]{\@ifundefined{#1}{\@namedef{#1}{}}\SkipToFmtEnd}
\usepackage{amstext}
\usepackage{amssymb}
\usepackage{stmaryrd}
\DeclareFontFamily{OT1}{cmtex}{}
\DeclareFontShape{OT1}{cmtex}{m}{n}
  {<5><6><7><8>cmtex8
   <9>cmtex9
   <10><10.95><12><14.4><17.28><20.74><24.88>cmtex10}{}
\DeclareFontShape{OT1}{cmtex}{m}{it}
  {<-> ssub * cmtt/m/it}{}

\DeclareFontShape{OT1}{cmtt}{bx}{n}
  {<5><6><7><8>cmtt8
   <9>cmbtt9
   <10><10.95><12><14.4><17.28><20.74><24.88>cmbtt10}{}
\DeclareFontShape{OT1}{cmtex}{bx}{n}
  {<-> ssub * cmtt/bx/n}{}

\newcommand{\Conid}[1]{\mathit{#1}}
\newcommand{\Varid}[1]{\mathit{#1}}
\newcommand{\anonymous}{\kern0.06em \vbox{\hrule\@width.5em}}

\renewcommand{\leq}{\leqslant}
\renewcommand{\geq}{\geqslant}
\usepackage{polytable}

\@ifundefined{mathindent}%
  {\newdimen\mathindent\mathindent\leftmargini}%
  {}%

\def\resethooks{%
  \global\let\SaveRestoreHook\empty
  \global\let\ColumnHook\empty}
\newcommand*{\savecolumns}[1][default]%
  {\g@addto@macro\SaveRestoreHook{\savecolumns[#1]}}
\newcommand*{\restorecolumns}[1][default]%
  {\g@addto@macro\SaveRestoreHook{\restorecolumns[#1]}}
\newcommand*{\aligncolumn}[2]%
  {\g@addto@macro\ColumnHook{\column{#1}{#2}}}

\resethooks

\newcommand{\onelinecommentchars}{\quad-{}- }
\newcommand{\commentbeginchars}{\enskip\{-}
\newcommand{\commentendchars}{-\}\enskip}

\newcommand{\visiblecomments}{%
  \let\onelinecomment=\onelinecommentchars
  \let\commentbegin=\commentbeginchars
  \let\commentend=\commentendchars}

\newcommand{\invisiblecomments}{%
  \let\onelinecomment=\empty
  \let\commentbegin=\empty
  \let\commentend=\empty}

\visiblecomments

\newlength{\blanklineskip}
\newcommand{\hsindent}[1]{\quad}
\let\hspre\empty
\let\hspost\empty

\EndFmtInput
\makeatother
\ReadOnlyOnce{polycode.fmt}%
\makeatletter

\newcommand{\hsnewpar}[1]%
  {{\parskip=0pt\parindent=0pt\par\vskip #1\noindent}}

\newcommand{\hscodestyle}{}

\newcommand{\sethscode}[1]%
  {\expandafter\let\expandafter\hscode\csname #1\endcsname
   \expandafter\let\expandafter\endhscode\csname end#1\endcsname}

  {\par\noindent
   \advance\leftskip\mathindent
   \hscodestyle
   \let\\=\@normalcr
   \(\pboxed}%
  {\endpboxed\)%
   \par\noindent
   \ignorespacesafterend}

  {\hsnewpar\abovedisplayskip
   \advance\leftskip\mathindent
   \hscodestyle
   \let\\=\@normalcr
   \(\pboxed}%
  {\endpboxed\)%
   \hsnewpar\belowdisplayskip
   \ignorespacesafterend}

\newcommand{\plainhs}{\sethscode{plainhscode}}
\plainhs

  {\hsnewpar\abovedisplayskip
   \advance\leftskip\mathindent
   \hscodestyle
   \let\\=\@normalcr
   \(\parray}%
  {\endparray\)%
   \hsnewpar\belowdisplayskip
   \ignorespacesafterend}

  {\parray}{\endparray}

  {\(\parray}{\endparray\)}

\def\codeframewidth{\arrayrulewidth}
\RequirePackage{calc}

  {\parskip=\abovedisplayskip\par\noindent
   \hscodestyle
   \arrayrulewidth=\codeframewidth
   \tabular{@{}|p{\linewidth-2\arraycolsep-2\arrayrulewidth-2pt}|@{}}%
   \hline\framedhslinecorrect\\{-1.5ex}%
   \let\endoflinesave=\\
   \let\\=\@normalcr
   \(\pboxed}%
  {\endpboxed\)%
   \framedhslinecorrect\endoflinesave{.5ex}\hline
   \endtabular
   \parskip=\belowdisplayskip\par\noindent
   \ignorespacesafterend}

\newcommand{\framedhslinecorrect}[2]%
  {#1[#2]}

  {\(\def\column##1##2{}%
   \let\>\undefined\let\<\undefined\let\\\undefined
   \newcommand\>[1][]{}\newcommand\<[1][]{}\newcommand\\[1][]{}%
   \def\fromto##1##2##3{##3}%
   }{\) }%

  {\let\orighscode=\hscode
   \let\origendhscode=\endhscode
   \def\endhscode{\def\hscode{\endgroup\def\@currenvir{hscode}\\}\begingroup}
   \orighscode\def\hscode{\endgroup\def\@currenvir{hscode}}}%
  {\origendhscode
   \global\let\hscode=\orighscode
   \global\let\endhscode=\origendhscode}%

\makeatother
\EndFmtInput
\ReadOnlyOnce{forall.fmt}%
\makeatletter

\let\HaskellResetHook\empty
\newcommand*{\AtHaskellReset}[1]{%
  \g@addto@macro\HaskellResetHook{#1}}
\newcommand*{\HaskellReset}{\HaskellResetHook}

\newcommand\hsforall{\global\let\hsdot=\hsperiodonce}
\newcommand*\hsperiodonce[2]{#2\global\let\hsdot=\hscompose}
\newcommand*\hscompose[2]{#1}

\AtHaskellReset{\global\let\hsdot=\hscompose}

\HaskellReset

\makeatother
\EndFmtInput

\begin{document}
\title{Proof Relevant Corecursive Resolution}
\author{Peng Fu\inst{1}\thanks{This author is supported by EPSRC grant EP/K031864/1.}, Ekaterina Komendantskaya\inst{1}, Tom Schrijvers\inst{2}, Andrew Pond\inst{1}\thanks{This author is supported by Carnegie Trust Scotland.}}
\institute{Computer Science, University of Dundee \and
           Department of Computer Science, KU Leuven }

\maketitle
\begin{abstract}
Resolution lies at the foundation of both logic programming and type class
context reduction in functional languages.
Terminating derivations by resolution have well-defined inductive meaning,
whereas some non-terminating derivations can be understood coinductively. 
Cycle detection is a popular method to capture a small subset of such
derivations.  We show that in fact cycle detection is a restricted form of
coinductive proof, in which the atomic formula forming the cycle plays the
r\^{o}le of coinductive hypothesis.

This paper introduces a heuristic method for obtaining richer coinductive
hypotheses in the form of Horn formulas. Our approach subsumes cycle detection
and gives coinductive meaning to a larger class of derivations. For this
purpose we extend resolution with Horn formula resolvents and  corecursive
evidence generation. We illustrate our method on non-terminating type class
resolution problems.

\textbf{Keywords:} Horn Clause Logic, Resolution, Corecursion, Haskell Type Class Inference, Coinductive Proofs.

\end{abstract}

\section{Introduction}
\label{intro}

Horn clause logic is a fragment of first-order logic known for its simple syntax, well-defined models, and efficient algorithms for automated proof search.
It is used in a variety of applications, from program verification~\cite{BjornerGMR15} to type inference in object-oriented programming languages ~\cite{AnconaLagorio11}.
Similar syntax and proof methods underlie type inference in functional programming languages~\cite{wadler1989make,Lammel:2005}.
For example, the following declaration
specifies equality class instances for pairs and integers in Haskell:

\begin{hscode}\SaveRestoreHook
\column{B}{@{}>{\hspre}l<{\hspost}@{}}%
\column{3}{@{}>{\hspre}l<{\hspost}@{}}%
\column{E}{@{}>{\hspre}l<{\hspost}@{}}%
\>[3]{}\mathbf{instance}\;\Conid{Eq}\;\Conid{Int}\;\mathbf{where}\mathbin{...}{}\<[E]%
\\
\>[3]{}\mathbf{instance}\;(\Conid{Eq}\;\Varid{x},\Conid{Eq}\;\Varid{y})\Rightarrow \Conid{Eq}\;(\Varid{x},\Varid{y})\;\mathbf{where}\mathbin{...}{}\<[E]%
\ColumnHook
\end{hscode}\resethooks
\noindent It corresponds to a Horn clause program $\Phi_{\ensuremath{\Conid{Pair}}}$ with two clauses $\kappa_{\ensuremath{\Conid{Int}}}$ and $\kappa_{\ensuremath{\Conid{Pair}}}$: 

\begin{hscode}\SaveRestoreHook
\column{B}{@{}>{\hspre}l<{\hspost}@{}}%
\column{3}{@{}>{\hspre}l<{\hspost}@{}}%
\column{17}{@{}>{\hspre}l<{\hspost}@{}}%
\column{E}{@{}>{\hspre}l<{\hspost}@{}}%
\>[3]{}\kappa_\Conid{Int}\mathbin{:}{}\<[17]%
\>[17]{}\Conid{Eq}\;\Conid{Int}{}\<[E]%
\\
\>[3]{}\kappa_\Conid{Pair}\mathbin{:}(\Conid{Eq}\;\Varid{x},\Conid{Eq}\;\Varid{y})\Rightarrow \Conid{Eq}\;(\Varid{x},\Varid{y}){}\<[E]%
\ColumnHook
\end{hscode}\resethooks



Horn clause logic uses SLD-resolution as an inference engine.
If a derivation for a given formula $A$ and a Horn clause program $\Phi$ terminates successfully with substitution $\theta$,
then $\theta A$ is logically entailed by $\Phi$, or $\Phi \vdash \theta A$. The search for a suitable $\theta$ reflects the problem-solving nature of SLD-resolution. 
When the unification algorithm underlying SLD-resolution is restricted to matching, resolution can be viewed as theorem proving:
the successful terminating derivations for $A$ using $\Phi$ will guarantee $\Phi \vdash A$.
For example, $ \ensuremath{\Conid{Eq}\;(\Conid{Int},\Conid{Int})} \leadsto  \ensuremath{\Conid{Eq}\;\Conid{Int}},  \ensuremath{\Conid{Eq}\;\Conid{Int}}   \leadsto  \ensuremath{\Conid{Eq}\;\Conid{Int}}  \leadsto \emptyset$.
Therefore, we have: 
$\Phi_{\ensuremath{\Conid{Pair}}} \vdash \ensuremath{\Conid{Eq}\;(\Conid{Int},\Conid{Int})}$.
For the purposes of this paper, we always assume resolution by term-matching. 

To emphasize the proof-theoretic meaning of resolution, we will record 
proof evidence alongside the derivation steps. For instance, \ensuremath{\Conid{Eq}\;(\Conid{Int},\Conid{Int})} is proven by applying the clauses $\kappa_{\ensuremath{\Conid{Pair}}}$ and $\kappa_{\ensuremath{\Conid{Int}}}$.
We denote this by $\Phi_{\ensuremath{\Conid{Pair}}} \vdash \ensuremath{\Conid{Eq}\;(\Conid{Int},\Conid{Int})}  \Downarrow
\kappa_{\ensuremath{\Conid{Pair}}}\ \kappa_{\ensuremath{\Conid{Int}}}\ \kappa_{\ensuremath{\Conid{Int}}}$. 


Horn clause logic can have inductive and coinductive interpretation, via the least and greatest fixed points of the \emph{consequence operator} $F_{\Phi}$. 
Given a Horn clause program $\Phi$, and a set $S$ containing (ground) formulas formed from the signature of $\Phi$,
$F_{\Phi}(S) = \{\sigma A\ | \ \sigma B_1, \ldots , \sigma B_n \in S \textrm{ and } B_1, \ldots B_n \Rightarrow A \in \Phi  \}$~\cite{Llo87}. 
Through the Knaster-Tarski construction, the least fixed point of this operator gives the set of all finite ground formulas \emph{inductively entailed} by $\Phi$.
Extending $S$ to include infinite terms,
 the greatest fixed point of $F_{\Phi}$ defines the set of all finite and infinite ground formulas \emph{coinductively entailed} by $\Phi$.

Inductively, SLD-resolution is sound: if $\Phi \vdash A$, then $A$ is inductively entailed by $\Phi$.
It is more difficult to characterise coinductive entailment computationally;
several alternative approaches 
exist~\cite{SimonBMG07,Llo87,KJ15}.
So far the most
popular solution has been to use cycle detection~\cite{SimonBMG07}: if a cycle
is found in a derivation for a formula $A$ and a Horn clause program $\Phi$,
then $A$ is coinductively entailed by $\Phi$.

Consider, as an example, the following Horn clause program $\Phi_{AB}$:
\begin{hscode}\SaveRestoreHook
\column{B}{@{}>{\hspre}l<{\hspost}@{}}%
\column{3}{@{}>{\hspre}l<{\hspost}@{}}%
\column{E}{@{}>{\hspre}l<{\hspost}@{}}%
\>[3]{}\kappa_\Conid{A}\mathbin{:}\Conid{B}\;\Varid{x}\Rightarrow \Conid{A}\;\Varid{x}{}\<[E]%
\\
\>[3]{}\kappa_\Conid{B}\mathbin{:}\Conid{A}\;\Varid{x}\Rightarrow \Conid{B}\;\Varid{x}{}\<[E]%
\ColumnHook
\end{hscode}\resethooks
It gives rise to an infinite derivation $ A\ x \leadsto B\ x  \leadsto A\ x \leadsto \ldots $. 
By noticing the cycle, 
we can conclude that (an instance) of $A\ x$ is coinductively entailed by $\Phi_{AB}$.
We are able to \emph{construct} a  proof evidence that reflects the circular nature of this derivation:  $\alpha = \kappa_A\ (\kappa_B \ \alpha)$.
This being a corecursive equation, we can represent it with the standard $\mu$ operator, $\mu \alpha.  \kappa_A\ (\kappa_B \ \alpha)$. Now we have
$\Phi_{AB} \vdash A\ x  \Downarrow \mu \alpha.  \kappa_A\ (\kappa_B \ \alpha)$.

According to Gibbons and Hutton~\cite{GH05} and inspired by Moss and Danner~\cite{MD97}, a corecursive program is defined to be a function whose range is a type defined recursively as the greatest solution of some equation (i.e. whose range is a coinductive type).
We can informally understand the Horn clause $\Phi_{AB}$ as the following Haskell data type declarations: 

\begin{hscode}\SaveRestoreHook
\column{B}{@{}>{\hspre}l<{\hspost}@{}}%
\column{3}{@{}>{\hspre}l<{\hspost}@{}}%
\column{20}{@{}>{\hspre}c<{\hspost}@{}}%
\column{20E}{@{}l@{}}%
\column{23}{@{}>{\hspre}l<{\hspost}@{}}%
\column{24}{@{}>{\hspre}l<{\hspost}@{}}%
\column{E}{@{}>{\hspre}l<{\hspost}@{}}%
\>[3]{}\mathbf{data}\;\Conid{B}\;\Varid{x}{}\<[20]%
\>[20]{}\mathrel{=}{}\<[20E]%
\>[23]{}\Conid{K_B}\;(\Conid{A}\;\Varid{x}){}\<[E]%
\\
\>[3]{}\mathbf{data}\;\Conid{A}\;\Varid{x}{}\<[20]%
\>[20]{}\mathrel{=}{}\<[20E]%
\>[24]{} \Conid{K_A}\;(\Conid{B}\;\Varid{x}){}\<[E]%
\ColumnHook
\end{hscode}\resethooks
\noindent So the corecursive evidence $\mu \alpha.  \kappa_A\ (\kappa_B \ \alpha)$ for $A\ x$ corresponds
to the corecursive program $(d\ :: A\ x) =\ K_A \ (K_B \ d)$.
In our case, the corecursive evidence $d$ is that function, and its range type $A \ x$ can be seen as a coinductive type. 

Corecursion also arises in type class inference.
Consider the following mutually recursive definitions of lists of even
and odd length in Haskell:
\begin{hscode}\SaveRestoreHook
\column{B}{@{}>{\hspre}l<{\hspost}@{}}%
\column{3}{@{}>{\hspre}l<{\hspost}@{}}%
\column{20}{@{}>{\hspre}c<{\hspost}@{}}%
\column{20E}{@{}l@{}}%
\column{23}{@{}>{\hspre}l<{\hspost}@{}}%
\column{24}{@{}>{\hspre}l<{\hspost}@{}}%
\column{E}{@{}>{\hspre}l<{\hspost}@{}}%
\>[3]{}\mathbf{data}\;\Conid{OddList}\;\Varid{a}{}\<[20]%
\>[20]{}\mathrel{=}{}\<[20E]%
\>[23]{}\Conid{OCons}\;\Varid{a}\;(\Conid{EvenList}\;\Varid{a}){}\<[E]%
\\
\>[3]{}\mathbf{data}\;\Conid{EvenList}\;\Varid{a}{}\<[20]%
\>[20]{}\mathrel{=}{}\<[20E]%
\>[24]{}\Conid{Nil}\mid \Conid{ECons}\;\Varid{a}\;(\Conid{OddList}\;\Varid{a}){}\<[E]%
\ColumnHook
\end{hscode}\resethooks
They give rise to \ensuremath{\Conid{Eq}} type class instance declarations that can be expressed
using the following Horn clause program $\Phi_{EvenOdd}$:  
\begin{hscode}\SaveRestoreHook
\column{B}{@{}>{\hspre}l<{\hspost}@{}}%
\column{3}{@{}>{\hspre}l<{\hspost}@{}}%
\column{E}{@{}>{\hspre}l<{\hspost}@{}}%
\>[3]{}\kappa_\Conid{Odd}\mathbin{:}(\Conid{Eq}\;\Varid{a},\Conid{Eq}\;(\Conid{EvenList}\;\Varid{a}))\Rightarrow \Conid{Eq}\;(\Conid{OddList}\;\Varid{a}){}\<[E]%
\\
\>[3]{}\kappa_\Conid{Even}\mathbin{:}(\Conid{Eq}\;\Varid{a},\Conid{Eq}\;(\Conid{OddList}\;\Varid{a}))\Rightarrow \Conid{Eq}\;(\Conid{EvenList}\;\Varid{a}){}\<[E]%
\ColumnHook
\end{hscode}\resethooks
When resolving the type class constraint \ensuremath{\Conid{Eq}\;(\Conid{OddList}\;\Conid{Int})}, Haskell's standard type class resolution diverges.
The state-of-the-art is to use cycle detection~\cite{Lammel:2005} to terminate otherwise infinite
derivations. 
Resolution for \ensuremath{\Conid{Eq}\;(\Conid{OddList}\;\Conid{Int})} exhibits a cycle on the atomic formula  \ensuremath{\Conid{Eq}\;(\Conid{OddList}\;\Conid{Int})}, thus the derivation can be terminated,
with corecursive evidence  $\mu \alpha .  \ensuremath{\kappa_\Conid{Odd}\;\kappa_\Conid{Int}\;(\kappa_\Conid{Even}\;\kappa_\Conid{Int}\;\alpha)}$.


The method of cycle detection is rather limited: there are many Horn clause
programs that have coinductive meaning, but do not give rise to
detectable cycles. For example, consider the program $\Phi_Q$:
\begin{hscode}\SaveRestoreHook
\column{B}{@{}>{\hspre}l<{\hspost}@{}}%
\column{3}{@{}>{\hspre}l<{\hspost}@{}}%
\column{15}{@{}>{\hspre}l<{\hspost}@{}}%
\column{E}{@{}>{\hspre}l<{\hspost}@{}}%
\>[3]{}\kappa_\Conid{S}\mathbin{:}(\Conid{Q}\;(\Conid{S}\;(\Conid{G}\;\Varid{x})),\Conid{Q}\;\Varid{x})\Rightarrow \Conid{Q}\;(\Conid{S}\;\Varid{x}){}\<[E]%
\\
\>[3]{}\kappa_\Conid{G}\mathbin{:}\Conid{Q}\;\Varid{x}\Rightarrow \Conid{Q}\;(\Conid{G}\;\Varid{x}){}\<[E]%
\\
\>[3]{}\kappa_\Conid{Z}\mathbin{:}{}\<[15]%
\>[15]{}\Conid{Q}\;\Conid{Z}{}\<[E]%
\ColumnHook
\end{hscode}\resethooks
It gives rise to the following derivation without cycling: 

\noindent $\underline{Q\ (S\ Z)} \leadsto Q\ Z, \underline{Q\ (S\ (G\ Z))} \leadsto Q\ Z, Q\ (G\ Z), \underline{Q\ (S \ (G\ (G\ Z)))}  \leadsto \ldots$.
When such derivations arise, we cannot terminate the derivation by cycle detection. 

Let us look at a similar situation for type classes. Consider a datatype-generic representation of perfect trees: a
nested datatype~\cite{bird1998nested}, with  fixpoint \ensuremath{\Conid{Mu}} of the higher-order functor \ensuremath{\Conid{HPTree}} \cite{johann2009haskell}.  
\begin{hscode}\SaveRestoreHook
\column{B}{@{}>{\hspre}l<{\hspost}@{}}%
\column{3}{@{}>{\hspre}l<{\hspost}@{}}%
\column{16}{@{}>{\hspre}c<{\hspost}@{}}%
\column{16E}{@{}l@{}}%
\column{19}{@{}>{\hspre}l<{\hspost}@{}}%
\column{E}{@{}>{\hspre}l<{\hspost}@{}}%
\>[3]{}\mathbf{data}\;\Conid{Mu}\;\Varid{h}\;\Varid{a}{}\<[16]%
\>[16]{}\mathrel{=}{}\<[16E]%
\>[19]{}\Conid{In}\;\{\mskip1.5mu \Varid{out}\mathbin{::}\Varid{h}\;(\Conid{Mu}\;\Varid{h})\;\Varid{a}\mskip1.5mu\}{}\<[E]%
\\[\blanklineskip]%
\>[3]{}\mathbf{data}\;\Conid{HPTree}\;\Varid{f}\;\Varid{a}\mathrel{=}\Conid{HPLeaf}\;\Varid{a}\mid \Conid{HPNode}\;(\Varid{f}\;(\Varid{a},\Varid{a})){}\<[E]%
\ColumnHook
\end{hscode}\resethooks
These two datatypes give rise to the following \ensuremath{\Conid{Eq}} type class instances.
\begin{hscode}\SaveRestoreHook
\column{B}{@{}>{\hspre}l<{\hspost}@{}}%
\column{3}{@{}>{\hspre}l<{\hspost}@{}}%
\column{5}{@{}>{\hspre}l<{\hspost}@{}}%
\column{16}{@{}>{\hspre}c<{\hspost}@{}}%
\column{16E}{@{}l@{}}%
\column{19}{@{}>{\hspre}c<{\hspost}@{}}%
\column{19E}{@{}l@{}}%
\column{20}{@{}>{\hspre}l<{\hspost}@{}}%
\column{22}{@{}>{\hspre}l<{\hspost}@{}}%
\column{31}{@{}>{\hspre}c<{\hspost}@{}}%
\column{31E}{@{}l@{}}%
\column{34}{@{}>{\hspre}l<{\hspost}@{}}%
\column{E}{@{}>{\hspre}l<{\hspost}@{}}%
\>[3]{}\mathbf{instance}\;\Conid{Eq}\;(\Varid{h}\;(\Conid{Mu}\;\Varid{h})\;\Varid{a})\Rightarrow \Conid{Eq}\;(\Conid{Mu}\;\Varid{h}\;\Varid{a})\;\mathbf{where}{}\<[E]%
\\
\>[3]{}\hsindent{2}{}\<[5]%
\>[5]{}\Conid{In}\;\Varid{x}\equiv \Conid{In}\;\Varid{y}{}\<[19]%
\>[19]{}\mathrel{=}{}\<[19E]%
\>[22]{}\Varid{x}\equiv \Varid{y}{}\<[E]%
\\[\blanklineskip]%
\>[3]{}\mathbf{instance}\;(\Conid{Eq}\;\Varid{a},\Conid{Eq}\;(\Varid{f}\;(\Varid{a},\Varid{a})))\Rightarrow \Conid{Eq}\;(\Conid{HPTree}\;\Varid{f}\;\Varid{a})\;\mathbf{where}{}\<[E]%
\\
\>[3]{}\hsindent{2}{}\<[5]%
\>[5]{}\Conid{HPLeaf}\;\Varid{x}{}\<[16]%
\>[16]{}\equiv {}\<[16E]%
\>[20]{}\Conid{HPLeaf}\;\Varid{y}{}\<[31]%
\>[31]{}\mathrel{=}{}\<[31E]%
\>[34]{}\Varid{x}\equiv \Varid{y}{}\<[E]%
\\
\>[3]{}\hsindent{2}{}\<[5]%
\>[5]{}\Conid{HPNode}\;\Varid{xs}{}\<[16]%
\>[16]{}\equiv {}\<[16E]%
\>[20]{}\Conid{HPNode}\;\Varid{ys}{}\<[31]%
\>[31]{}\mathrel{=}{}\<[31E]%
\>[34]{}\Varid{xs}\equiv \Varid{ys}{}\<[E]%
\\
\>[3]{}\hsindent{2}{}\<[5]%
\>[5]{}\anonymous {}\<[16]%
\>[16]{}\equiv {}\<[16E]%
\>[20]{}\anonymous {}\<[31]%
\>[31]{}\mathrel{=}{}\<[31E]%
\>[34]{}\Conid{False}{}\<[E]%
\ColumnHook
\end{hscode}\resethooks
The corresponding Horn clause program $\Phi_{HPTree}$ consists of $\Phi_{Pair}$ and the following two clauses :
\begin{hscode}\SaveRestoreHook
\column{B}{@{}>{\hspre}l<{\hspost}@{}}%
\column{3}{@{}>{\hspre}l<{\hspost}@{}}%
\column{E}{@{}>{\hspre}l<{\hspost}@{}}%
\>[3]{}\kappa_\Conid{Mu}\mathbin{:}\Conid{Eq}\;(\Varid{h}\;(\Conid{Mu}\;\Varid{h})\;\Varid{a})\Rightarrow \Conid{Eq}\;(\Conid{Mu}\;\Varid{h}\;\Varid{a}){}\<[E]%
\\
\>[3]{}\kappa_\Conid{HPTree}\mathbin{:}(\Conid{Eq}\;\Varid{a},\Conid{Eq}\;(\Varid{f}\;(\Varid{a},\Varid{a})))\Rightarrow \Conid{Eq}\;(\Conid{HPTree}\;\Varid{f}\;\Varid{a}){}\<[E]%
\ColumnHook
\end{hscode}\resethooks
\noindent The type class resolution for \ensuremath{\Conid{Eq}\;(\Conid{Mu}\;\Conid{HPTree}\;\Conid{Int})} cannot be terminated by cycle 
detection. Instead we get a context reduction overflow error in the Glasgow Haskell Compiler, even if we just compare two finite 
data structures of the type \ensuremath{\Conid{Mu}\;\Conid{HPTree}\;\Conid{Int}}. 

To find a solution to the above problems, 
let us view infinite resolution from 
the perspective of coinductive proof in the Calculus of Coinductive Constructions
\cite{coquand1994infinite,gimenez1996}.
There, in order to prove a proposition $F$ from the assumptions $F_1, .., F_n$,
the proof may involve not only natural deduction and lemmas, but also $F$, provided the use of $F$ is \emph{guarded}.
We could say that the existing cycle detection methods treat the atomic formula forming a cycle
as a
\emph{coinductive hypothesis}.
We can equivalently describe the above-explained derivation for $\Phi_{AB}$ in the following terms:
when a cycle with a formula \ensuremath{\Conid{A}\;\Varid{x}} is found in the derivation, $\Phi_{AB}$ gets extended with a coinductive hypothesis $\alpha:\  \ensuremath{\Conid{A}\;\Varid{x}} $.
So to prove \ensuremath{\Conid{A}\;\Varid{x}} coinductively, we would need to apply the clause $\kappa_{A}$ first, and then clause $\kappa_B$, finally apply the coinductive hypothesis.
The resulting proof witness is $\mu \alpha. \  \kappa_A \ (\kappa_B\ \alpha)$.



The next logical step we can make is to use the above formalism to extend the syntax of the coinductive hypotheses. %
While cycle detection only uses atomic formulas as coinductive hypotheses, we can try to generalise the syntax of coinductive hypotheses to full Horn formulas. 


For example, for program $\Phi_Q$, we could prove a lemma
\ensuremath{\Varid{e}\mathbin{:}\Conid{Q}\;\Varid{x}\Rightarrow \Conid{Q}\;(\Conid{S}\;\Varid{x})} coinductively, 
which would allow us to form finite derivation for \ensuremath{\Conid{Q}\;(\Conid{S}\;\Conid{Z})}, which is described by $(e \ \kappa_Z)$.
The proof of \ensuremath{\Varid{e}\mathbin{:}\Conid{Q}\;\Varid{x}\Rightarrow \Conid{Q}\;(\Conid{S}\;\Varid{x})} is of a coinductive nature: if we first assume
 $\alpha : \ensuremath{\Conid{Q}\;\Varid{x}\Rightarrow \Conid{Q}\;(\Conid{S}\;\Varid{x})}$ and $\alpha_1 : Q \ C$, then all we need to show
is \ensuremath{\Conid{Q}\;(\Conid{S}\;\Conid{C})}.\footnote{Note that here \ensuremath{\Conid{C}} is an eigenvariable.}
To show \ensuremath{\Conid{Q}\;(\Conid{S}\;\Conid{C})}, we apply $\kappa_S$, which gives us \ensuremath{\Conid{Q}\;\Conid{C}}, \ensuremath{\Conid{Q}\;(\Conid{S}\;(\Conid{G}\;\Conid{C}))}. We first discharge \ensuremath{\Conid{Q}\;\Conid{C}}
 with $\alpha_1$ and then  apply the coinductive hypothesis $\alpha$ which 
 yields \ensuremath{\Conid{Q}\;(\Conid{G}\;\Conid{C})}, and can be proved with $\kappa_G$ and $\alpha_1$. So we
have obtained a coinductive proof for $e$, which is $\mu \alpha .\lambda
\alpha_1 . \kappa_S\ (\alpha\ (\kappa_G\ \alpha_1))\ \alpha_1$. We can apply
similar reasoning 
to show that $\Phi_{\ensuremath{\Conid{HPTree}}} \vdash \ensuremath{\Conid{Eq}\;(\Conid{Mu}\;\Conid{HPTree}\;\Conid{Int})} \  \Downarrow (\mu \alpha . \lambda \alpha_1 . \kappa_{\ensuremath{\Conid{Mu}}}\ (\kappa_{\ensuremath{\Conid{HPTree}}}\ \alpha_1\ (\alpha\  (\kappa_{\ensuremath{\Conid{Pair}}}\ \alpha_1\ \alpha_1)))) \ \kappa_{\ensuremath{\Conid{Int}}}$
using the coinductively proved lemma \ensuremath{\Conid{Eq}\;\Varid{x}\Rightarrow \Conid{Eq}\;(\Conid{Mu}\;\Conid{HPTree}\;\Varid{x})} \footnote{The proof term can be type-checked with polymorphic recursion.}.

To formalise the above intuitions, we need to solve several technical problems.

 
  \textit{1. How to generate suitable lemmas?} We propose to observe a more general notion of a loop invariant than a cycle in the non-terminating resolution. 
  In Section \ref{inv} we devise a heuristic method to identify potential loops in the resolution tree and extract \textit{candidate lemmas} in Horn clause form. 

\noindent Generally, it is very challenging to develop a practical method for generating
candidate lemmas based on loop analysis, since the admissibility of a loop in reduction is a semi-decidable problem~\cite{zantema1996non}.

 \textit{2. How to enrich resolution to allow coinductive proofs for Horn formulas? and how to formalise the corecursive proof evidence construction?} Coinductive proofs involve not 
only applying the axioms, but also modus ponens and generalization. Therefore, the resolution
mechanism will have to be extended in order to support such automation.   

\noindent In Section \ref{corec}, we introduce proof relevant \emph{corecursive resolution} -- a calculus that extends the standard resolution rule with two further rules:
one allows us to resolve Horn formula queries, and the other to construct corecursive proof evidence  for non-terminating resolution.


  \emph{3. How to give an operational semantics to the evidence produced by corecursive resolution of Section~\ref{corec}?}
  In particular, we need to show the correspondence between corecursive evidence and resolution seen as infinite reduction.
  In Section~\ref{ob-equiv}, we prove that for every non-terminating resolution resulting from a \textit{simple loop}, a coinductively provable
 candidate lemma can be obtained and its evidence is \textit{observationally equivalent} to the non-terminating resolution process.

\noindent In type class inference, the proof evidence has computational meaning, i.e. the evidence will be run as a program. So the corecursive evidence should be able to recover the original infinite resolution trace.


In Sections~\ref{relwork} and~\ref{fwork} we survey the related work, explain the limitations of our method and conclude the paper. 
We have implemented our method of candidate lemma generation based on loop analysis and corecursive
resolution, and incorporated it in the type class inference process of a simple
functional language. 
Additional examples 
and implementation information are provided in the extended version. 

\section{Preliminaries: Resolution with Evidence}
\label{preliminary}

This section provides a standard formalisation of resolution with evidence
together with two derived forms: a
\emph{small-step} variant of resolution and a reification of resolution in a
resolution tree.

We consider the following syntax.
\begin{definition}[Basic syntax]
{\small
\[\begin{array}{l@{\hspace{5mm}}l@{\hspace{5mm}}c@{\hspace{5mm}}l}
  \text{Term} & t \ & ::= & \ x \ | \ K\ | \ t\ t' \\
  \text{Atomic Formula} & A, B, C, D\ & ::= & \ P\ t_1 \ ...\ t_n \\
  \text{Horn Formula} & H \ & ::= & \ B_1,..., B_n \Rightarrow A \\
  \text{Proof/Evidence} & e \ & ::= & \ \kappa \ |\ e \ e' \\
  \text{Axiom Environment} & \Phi \ & ::= & \cdot \ | \ \Phi, (\kappa : H) 
\end{array} \]}
\end{definition}

We consider first-order applicative terms, where $K$ stands for some 
constant symbol. Atomic formulas are predicates
on terms, and Horn formulas are defined as usual. We assume that 
all variables $x$ in Horn formulas are implicitly universally quantified.
There are no \textit{existential variables} in the Horn formulas, i.e.,
$\bigcup_i \mathrm{FV}(B_i) \subseteq \mathrm{FV}(A)$ for $B_1,\ldots,B_n
\Rightarrow A$. 
The axiom environment $\Phi$ is a set of Horn formulas labelled with distinct
evidence constants $\kappa$. Evidence terms $e$ are made of evidence
constants $\kappa$ and their applications. Finally, 
we often use $\underline{A}$ to abbreviate $A_1,..., A_n$ when the number $n$ is
unimportant. 



The above syntax can be used to model the Haskell type class setting as follows.
Terms denote Haskell types like \ensuremath{\Conid{Int}} or \ensuremath{(\Varid{x},\Varid{y})}, and atomic formulas
denote Haskell type class constraints on types like \ensuremath{\Conid{Eq}\;(\Conid{Int},\Conid{Int})}. Horn formulas
correspond to the type-level information of type class instances. 



Our evidence $e$ models type class dictionaries, following Wadler and Blott's
dictionary-passing elaboration of type classes \cite{wadler1989make}.  In particular the constants
$\kappa$ refer to dictionaries that capture the term-level information of type
class instances, i.e., the implementations of the type class methods.  Evidence
application $(e\ e')$ accounts for dictionaries that are parametrised by other
dictionaries.  Horn formulas in turn represent type class instance
declarations. The axiom environment $\Phi$ corresponds to Haskell's global
environment of type class instances. Note that the treatment of type class instance declaration
and their corresponding evidence construction here
are based on our own understanding of many related works (\cite{Jones97,jones2003qualified}), which are also discussed in Section \ref{relwork}.  



In order to define resolution  together with evidence generation, we use resolution judgement $\Phi \vdash A \Downarrow e$ to state that the
atomic formula $A$ is entailed by the axioms $\Phi$, and that the proof
term $e$ witnesses this entailment. It is defined by means of the following
 inference rule.
\begin{definition}[Resolution]\label{def:resolution}
\fbox{$\Phi \vdash A \Downarrow e$}
\label{context-red}
{\small
\[\begin{array}{c}
  \infer[\text{if}~(\kappa : B_1,..., B_n \Rightarrow A) \in \Phi]
    {\Phi \vdash \sigma A \Downarrow \kappa~e_1 \cdots e_n}
    { \Phi \vdash \sigma B_1 \Downarrow e_1 \quad \cdots \quad \Phi \vdash \sigma B_n \Downarrow e_n 
      } 
  \end{array}
\]}
\end{definition}
 Using this definition we can show $\Phi_{\ensuremath{\Conid{Pair}}} \vdash \ensuremath{\Conid{Eq}\;(\Conid{Int},\Conid{Int})} \Downarrow \ensuremath{\kappa_\Varid{Pair}\;\kappa_\Varid{Int}\;\kappa_\Varid{Int}}$.

In case  resolution is diverging, it is often more convenient to consider
a \emph{small-step} resolution judgement (in analogy to the small step
operational semantics) that performs one resolution step at a time and allows
us to observe the intermediate states.

The basic idea is to rewrite the initial query $A$ step by step into its
evidence~$e$. This involves \emph{mixed terms} on the way that consist partly
of evidence, and partly of formulas that are not yet resolved.
\begin{definition}[Mixed Terms]
{\small
\[\begin{array}{l@{\hspace{5mm}}l@{\hspace{5mm}}c@{\hspace{5mm}}l}
  \text{Mixed term} & q \ & ::= & \ A \ |\ \kappa \ | \ q \ q' \\
  \text{Mixed term context} & \mathcal{C} \ & ::= & \ \bullet \ | \ \mathcal{C}\ q \ | \ q\ \mathcal{C}\\
\end{array}
\]}
\end{definition}
At the same time we have defined mixed term contexts $\mathcal{C}$ as mixed
terms with a hole $\bullet$, where $\mathcal{C}[q]$ substitutes the hole with $q$ in the usual way.


\begin{definition}[Small-Step Resolution]
\fbox{$\Phi \vdash q \rightarrow q'$}
\label{context-red}
{\small
\[\begin{array}{c}
  \infer[\text{if}~(\kappa : \underline{B} \Rightarrow A) \in \Phi]
    {\Phi \vdash \mathcal{C}[\sigma A] \rightarrow \mathcal{C}[\kappa~\sigma\underline{B}]}
    {  } 
  \end{array}
\]}
\end{definition}
For instance, we resolve \ensuremath{\Conid{Eq}\;(\Conid{Int},\Conid{Int})} in three small
steps: $\Phi_{\ensuremath{\Conid{Pair}}} \vdash \ensuremath{\Conid{Eq}\;(\Conid{Int},\Conid{Int})} \rightarrow \ensuremath{\kappa_\Varid{Pair}\;(\Conid{Eq}\;\Conid{Int})\;(\Conid{Eq}\;\Conid{Int})}$,
       $\Phi_{\ensuremath{\Conid{Pair}}} \vdash \ensuremath{\kappa_\Varid{Pair}\;(\Conid{Eq}\;\Conid{Int})\;(\Conid{Eq}\;\Conid{Int})} \rightarrow \ensuremath{\kappa_\Varid{Pair}\;\kappa_\Varid{Int}\;(\Conid{Eq}\;\Conid{Int})}$ and
       $\Phi_{\ensuremath{\Conid{Pair}}} \vdash \ensuremath{\kappa_\Varid{Pair}\;\kappa_\Varid{Int}\;(\Conid{Eq}\;\Conid{Int})} \rightarrow \ensuremath{\kappa_\Varid{Pair}\;\kappa_\Varid{Int}\;\kappa_\Varid{Int}}$.
We write  $\Phi \vdash q \to^{*} q'$ to denote the transitive closure of
small-step resolution.


The following theorem formalizes the intuition that resolution and small-step resolution coincide.
\begin{theorem}
\label{func:eq}
 $\Phi \vdash A \Downarrow e$ iff $ \Phi \vdash A \to^* e$.
\end{theorem}

The proof tree for a judgement $\Phi \vdash A \Downarrow e$ is called a
\emph{resolution tree}.  It conveniently records the history of resolution and,
for instance, it is easy to observe the ancestors of a node. This last feature
is useful for our heuristic loop invariant analysis in Section \ref{inv}.    
 
Our formalisation of trees in general is as follows: We use $w, v$ to denote
positions $\langle k_1, k_2, ..., k_n \rangle$ in a tree, where $k_i \geq 1 $
for $ 1 \leq i \leq n $.  Let $\epsilon$ denote the empty position or
\emph{root}. We also define $\langle k_1, k_2, ..., k_n \rangle \cdot i =
\langle k_1, k_2, ..., k_n, i\rangle $ and $\langle k_1, k_2, ..., k_n \rangle
+ \langle l_1,..., l_m \rangle = \langle k_1, k_2, ..., k_n, l_1, ..., l_m
\rangle$. We write $w > v$ if there exists a non-empty $v'$ such that $w = v + v'$. For a tree $T$, $T(w)$ refers to the node at position $w$, and
$T(w, i)$ refers to the edge between $T(w)$ and $T(w\cdot i)$. We use $\Box$ as a special proposition to denote
success.

 Resolution trees are defined as follows, note that they are a special case of \textit{rewriting
 trees}~\cite{KomendantskayaEtAl15,KJ15}: 
\begin{definition}[Resolution Tree]
The resolution tree for atomic formula $A$ is
a tree $T$ satisfying:
\begin{itemize}
\item $T(\epsilon) =  A$.
\item $T(w \cdot i) = \sigma B_i$ and $T(w, i) = \kappa^i$ with $i \in \{1
,..., n\}$ if $T(w) = \sigma D$ and $(\kappa : B_1, ..., B_n \Rightarrow D) \in
\Phi$. When $n = 0$, we write $T(w \cdot i) = \Box$ and $T(w, i) = \kappa$ for
any $i > 0$.  
\end{itemize}
\end{definition}


In general, the resolution tree can be infinite, this means that resolution is
non-terminating, which we denote as $\Phi \vdash A \Uparrow$.
Figure~\ref{fig:ex2} shows a finite fragment of the infinite resolution tree for $\Phi_{\ensuremath{\Conid{HPTree}}} \vdash \ensuremath{\Conid{Eq}\;(\Conid{Mu}\;\Conid{HPTree}\;\Conid{Int})} \Uparrow$.



\begin{figure}[tbp]
  \makebox[\textwidth]{

\begin{tikzpicture}[every tree node/.style={},
   level distance=1.20cm,sibling distance=1.00cm, 
   edge from parent path={(\tikzparentnode) -- (\tikzchildnode)}, font=\scriptsize]
\Tree [.\node{\underline{\ensuremath{\Conid{Eq}\;(\Conid{Mu}\;\Conid{HPTree}\;\Conid{Int})}}};
        \edge node[auto=right]{$\ensuremath{\kappa_{\Varid{Mu}}}^1$};      
         [.\node{\ensuremath{\Conid{Eq}\;(\Conid{HPTree}\;(\Conid{Mu}\;\Conid{HPTree})\;\Conid{Int})}};  
        \edge node[auto=right]{$\ensuremath{\kappa_{\Varid{HPTree}}}^1$};      
         [.\node{\ensuremath{\Conid{Eq}\;\Conid{Int}}};  
           \edge node[auto=right]{$\ensuremath{\kappa_\Varid{Int}}$};      
         [.\node{$\Box$};  
         ]
         ]
         \edge node[auto=left]{$\ensuremath{\kappa_{\Varid{HPTree}}}^2$};      
         [.\node{\underline{\ensuremath{\Conid{Eq}\;(\Conid{Mu}\;\Conid{HPTree}\;(\Conid{Int},\Conid{Int}))}}};  
           \edge node[auto=right]{$\ensuremath{\kappa_{\Varid{Mu}}}^1$};
                 [.\node{\ensuremath{\Conid{Eq}\;(\Conid{HPTree}\;(\Conid{Mu}\;\Conid{HPTree})\;(\Conid{Int},\Conid{Int}))}};
                   \edge node[auto=right]{$\ensuremath{\kappa_{\Varid{HPTree}}}^1$};
                   [.\node{...};]
                   \edge node[auto=left]{$\ensuremath{\kappa_{\Varid{HPTree}}}^2$};
                   [.\node{...};]
                 ]
         ]
         ]
         ]
\end{tikzpicture}
}
	\caption{The infinite resolution tree for $\Phi_{\ensuremath{\Conid{HPTree}}} \vdash \ensuremath{\Conid{Eq}\;(\Conid{Mu}\;\Conid{HPTree}\;\Conid{Int})} \Uparrow$}
	\label{fig:ex2}
\end{figure}
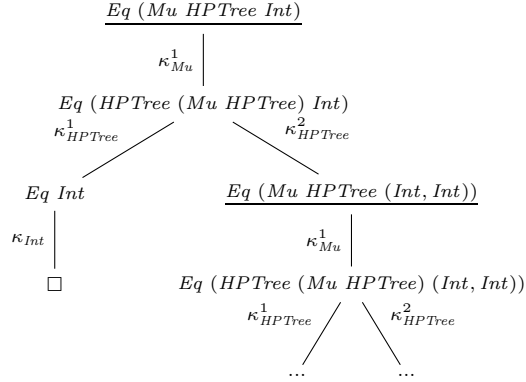

We note that Definitions~~\ref{def:resolution} and~\ref{context-red} describe a special case of SLD-resolution in which unification taking place in derivations is restricted to term-matching.
This restriction is motivated by two considerations. The first one comes directly from the application area of our results: type class resolution uses exactly this restricted version of SLD-resolution.
The second reason is of more general nature. As discussed in detail in~\cite{FuK15,KJ15}, SLD-derivations restricted to term-matching reflect the \emph{theorem proving} content of a proof by SLD-resolution. That is,
if $A$ can be derived  from $\Phi$ by SLD-resolution with term-matching only, then $A$ is inductvely entailed by $\Phi$. If, on the other hand, $A$ is derived from $\Phi$ by SLD-resolution with unification and computes a substitution $\sigma$,
then $\sigma A$ is inductively entailed by $\Phi$. In this sense, SLD-resolution with unification additionally has a \emph{problem-solving} aspect.
In developing proof-theoretic approach to resolution here, we thus focus on resolution by term-matching.


The resolution rule of Definition~\ref{def:resolution} resembles the definition of the consequence operator~\cite{Llo87} used to
define declarative semantics of Horn clause Logic.
In fact,  the forward and backward closure of the  rule of Definition~\ref{def:resolution} can be directly used to construct the usual least and greatest Herbrand models
for Horn clause logic, as shown in~\cite{KJ15}. There, it was also shown that SLD-resolution by term-matching is sound but incomplete relative to the least Herbrand models.


\section{Candidate Lemma Generation}
\label{inv}

This section explains how we generate candidate lemma from a
 potentially infinite resolution tree. Based on Paterson's condition we obtain
a finite pruned approximation (Definition~\ref{closed}) of this resolution tree.
Anti-unification on this approximation yields an abstract atomic formula and
 the corresponding abstract approximated resolution tree. It is from this abstract
tree that we read off the candidate lemma (Definition~\ref{invariant}).

We use $\Sigma(A)$ and $\mathrm{FVar}(A)$ to denote the
multi-sets of respectively function symbols and variables in $A$.

\begin{definition}[Paterson's Condition]
  For $(\kappa : \underline{B} \Rightarrow A) \in \Phi$, we say $\kappa$
satisfies Paterson's condition if $(\Sigma(B_i) \cup \mathrm{FVar}(B_i))
\subset (\Sigma(A) \cup \mathrm{FVar}(A))$ for each $B_i$. 
\end{definition}

Paterson's condition is used in Glasgow Haskell Compiler to enforce termination of context
reduction \cite{SulzmannDJS07}. In this paper, we use it as a practical
criterion to detect problematic instance declarations. Any declarations that do
not satisfy the condition could potentially introduce diverging behavior in
the resolution tree. 


If $\kappa : A_1,..., A_n \Rightarrow B$, then we have $\kappa^i : A_i \Rightarrow B$ for projection on index $i$.

\begin{definition}[Critical Triple]
\label{loop}
Let $v = (w \cdot i) + v'$ for some $v'$. A critical triple in $T$ is a triple $\langle \kappa^i, T(w), T(v) \rangle$ such that $T(v, i) = T(w, i) = \kappa^i$, and $\kappa^i$ does not satisfy Paterson's condition. 
\end{definition}

We will omit $\kappa^i$ from the triple and write $\langle T(w), T(v)
\rangle$ when it is not important. Intuitively, it means the nodes $T(w)$ and $T(v)$ are
using the same problematic projection $\kappa^i$, which could give rise to infinite resolution. 

The absence of a critical triple in a resolution tree means
that it has to be finite \cite{SulzmannDJS07}, while the presence of
a critical triple only means that the tree is \textit{possibly} infinite. In
general the infiniteness of a resolution tree is undecidable and
the critical triples provide a convenient over-approximation.

\begin{definition}[Closed Subtree]
  \label{closed}
  A closed subtree $T$ is a subtree of a resolution tree such that for all leaves $T(v) \not = \Box$,
the root $T(\epsilon)$ and $T(v)$ form a critical triple. 
\end{definition}

The critical triple in Figure \ref{fig:ex2} is formed by the underlined nodes. The
closed subtree in that figure is the subtree without the infinite branch below node
\ensuremath{\Conid{Eq}\;(\Conid{Mu}\;\Conid{HPTree}\;(\Conid{Int},\Conid{Int}))}. A closed subtree can intuitively be
understood as a finite approximation of an infinite resolution tree. We use it
as the basis for generating candidate lemma by means of 
anti-unification \cite{plotkin1970note}. 


\begin{definition}[Anti-Unifier]
We define the least general anti-unifier of atomic formulas $A$ and $B$ (denoted by $A \sqcup B$) and 
the least general anti-unifier of the terms $t$ and $t'$ (denoted by $t \sqcup t'$) as: 

  \begin{itemize}
  \item $P\ t_1\ ..., t_n\ \sqcup P\ t_1'\ ..., t_n' = P\ (t_1 \sqcup t_1')\ ...\ (t_n \sqcup t_n')$
  \item $K\ t_1\ ...\ t_n \sqcup K \ t_1'\ ...\ t_n' = K\ (t_1\sqcup t_1')\ ...\ (t_n\sqcup t_n')$
  \item Otherwise, $A \sqcup B = \phi(A,B), t \sqcup t' = \phi(t,t')$, where $\phi$ is an injective function from a pair of terms (atomic formulas) to a set of fresh variables. 
  \end{itemize}
\end{definition}

Anti-unification allows us to extract the common pattern from different ground atomic formulas.

\begin{definition}[Abstract Representation]
  \label{abs:rep}
  Let $\langle T(\epsilon), T(v_1)\rangle,..., \langle T(\epsilon), T(v_n) \rangle$ be all the critical triples in a closed subtree $T$. 
  Let  $C = T(\epsilon) \sqcup T(v_1) \sqcup ... \sqcup T(v_n)$, then the abstract
  representation $T'$ of the closed subtree $T$ is a tree such that:
  \begin{itemize}
  \item $T'(\epsilon) = C$
  \item $T'(w \cdot i) = \sigma B_i$ and $T'(w, i) = \kappa^i$ with $i \in \{1
,..., n\}$ if $T'(w) = \sigma D$ and $(\kappa : B_1, ..., B_n \Rightarrow D) \in
\Phi$. When $n = 0$, we write $T'(w \cdot i) = \Box$ and $T'(w, i) = \kappa$ for
any $i > 0$.  
  \item $T'(w)$ is undefined if $w > v_i$ for some $1 \leq i \leq n$.
  \end{itemize}

\end{definition}

The abstract representation unfolds the anti-unifier of all the critical
triples. Thus the  abstract representation can always be embedded into the
original closed subtree. It is an abstract form of the closed subtree, and we
can extract the candidate lemma from the abstract representation.  
  
\begin{definition}[Candidate Lemma]
  \label{invariant}
  Let $T$ be an abstract representation of a closed subtree, then the candidate lemma
  induced by this abstract representation is $T(v_1),..., T(v_n) \Rightarrow T(\epsilon)$, 
  where the $T(v_i)$ are all the leaves for which $T(v_i) \Rightarrow T(\epsilon)$ satisfies Paterson's
  condition.
\end{definition}

Figure \ref{fig:ex3} shows the abstract representation of the closed subtree of Figure \ref{fig:ex2}.
We read off the candidate lemma as \ensuremath{\Conid{Eq}\;\Varid{x}\Rightarrow \Conid{Eq}\;(\Conid{Mu}\;\Conid{HPTree}\;\Varid{x})}. 

\begin{figure}[tbp]
  \makebox[\textwidth]{

\begin{tikzpicture}[every tree node/.style={},
   level distance=1.20cm,sibling distance=1.00cm, 
   edge from parent path={(\tikzparentnode) -- (\tikzchildnode)}, font=\scriptsize]
\Tree [.\node{$\underline{\ensuremath{\Conid{Eq}\;(\Conid{Mu}\;\Conid{HPTree}\;\Varid{x})}}$};
        \edge node[auto=right]{$\ensuremath{\kappa_{\Varid{Mu}}}^1$};      
         [.\node{\ensuremath{\Conid{Eq}\;(\Conid{HPTree}\;(\Conid{Mu}\;\Conid{HPTree})\;\Varid{x})}};  
        \edge node[auto=right]{$\ensuremath{\kappa_{\Varid{HPTree}}}^1$};      
         [.\node{\ensuremath{\Conid{Eq}\;\Varid{x}}};  
         ]
         \edge node[auto=left]{$\ensuremath{\kappa_{\Varid{HPTree}}}^2$};      
         [.\node{\underline{\ensuremath{\Conid{Eq}\;(\Conid{Mu}\;\Conid{HPTree}\;(\Varid{x},\Varid{x}))}}};  
         ]
         ]
         ]
         ]
\end{tikzpicture}
}
	\caption{The abstract representation of the closed subtree of Figure~\ref{fig:ex2}}
	\label{fig:ex3}
\end{figure}
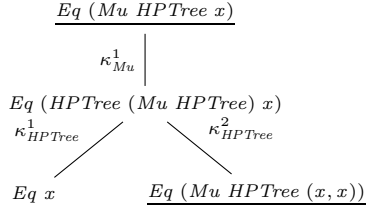

The candidate lemma plays a double role. Firstly, it allows us to
construct a finite resolution tree. For example, we know that \ensuremath{\Conid{Eq}\;(\Conid{Mu}\;\Conid{HPTree}\;\Conid{Int})} gives rise to infinite tree with the axiom environment $\Phi_{\ensuremath{\Conid{HPTree}}}$. However, a finite
tree can be constructed with \ensuremath{\Conid{Eq}\;\Varid{x}\Rightarrow \Conid{Eq}\;(\Conid{Mu}\;\Conid{HPTree}\;\Varid{x})},
since it reduces \ensuremath{\Conid{Eq}\;(\Conid{Mu}\;\Conid{HPTree}\;\Conid{Int})} to \ensuremath{\Conid{Eq}\;\Conid{Int}}, which succeeds trivially
with \ensuremath{\kappa_\Varid{Int}}. Next we show how to prove the candidate lemma coinductively, and
such proofs will encapsulate the
infinite aspect of the resolution tree. Since an infinite resolution tree gives
rise to infinite evidence, the finite proof of the lemma has to be
coinductive. We discuss such evidence construction in detail in Section~\ref{corec} and Section~\ref{ob-equiv}.

%

\section{Corecursive Resolution}
\label{corec}



In this section, we extend the definition of resolution from  Section~\ref{preliminary}
by introducing two additional rules: one to handle coinductive proofs, and another -- to allow Horn formula goals, rather than atomic goals, in the derivations.
We call the resulting calculus \textit{corecursive resolution}.

\begin{definition}[Extended Syntax]
\label{syntax}
{\small
\[
\begin{array}{l@{\hspace{5mm}}l@{\hspace{5mm}}c@{\hspace{5mm}}l}
    \text{Proof/Evidence} &  e    & ::= & \kappa ~\mid~ e \ e' ~\mid~ \textcolor{blue}{\alpha} ~\mid~ \textcolor{blue}{\lambda \alpha . e} ~\mid~ \textcolor{blue}{\mu \alpha. e} \\
    \text{Axiom Environment}  &  \Phi & ::= & \cdot ~\mid~ \Phi, (\textcolor{blue}{e} : H)
\end{array}
\]}
\end{definition}

To support coinductive proofs, we extend the syntax of evidence with functions
$\lambda \alpha . e$, variables $\alpha$ and fixed point $\mu \alpha . e$ (which models the corecursive equation $\alpha = e$ in which $\alpha$ occurs in $e$). Also
we allow the Horn clauses $H$ in the axiom environment $\Phi$ to be supported by any
form of evidence $e$ (and not necessarily by constants $\kappa$).


\begin{definition}[Corecursive Resolution]\label{def:cresolution}
  The following judgement for corecursive resolution extends the resolution in Definition~\ref{def:resolution}. {\small
\[\begin{array}{c}
  \infer[\text{if}~(e : B_1,..., B_m \Rightarrow A) \in \Phi]
    {\Phi \vdash \sigma A \Downarrow e~e_1 \cdots e_n}
    { \Phi \vdash \sigma B_1 \Downarrow e_1 \quad \cdots \quad \Phi \vdash \sigma B_n \Downarrow e_n 
      } 
\\ \\
    \infer[(\textsc{Mu})]
    {\Phi \vdash \underline{A} \Rightarrow B \Downarrow \mu \alpha . e }
    {\Phi, (\alpha : \underline{A} \Rightarrow B) \vdash \underline{A} \Rightarrow B \Downarrow e & \mathrm{HNF}(e)} 
\quad\quad
    \infer[(\textsc{Lam})]
    {\Phi \vdash \underline{A} \Rightarrow B \Downarrow \lambda \underline{\alpha} . e}
    {\Phi, (\underline{\alpha} : \underline{A}) \vdash  B \Downarrow e}
  \end{array}
\]}

\end{definition}
Note that $\mathrm{HNF}(e)$ means $e$ has to be in \textit{head normal form}
$\lambda \underline{\alpha}.\kappa\ \underline{e}$. This requirement is essential to ensure
the corecursive evidence satisfies the \textit{guardedness} condition.\footnote{See the extended version for a detailed discussion.}  The \textsc{Lam} rule implicitly assumes the treatment of \textit{eigenvariables},
i.e. we instantiate all the free variables in $\underline{A} \Rightarrow B$
with fresh term constants. 

We implicitly assume that axiom environments are well-formed. 
\begin{definition}[Well-formedness of environment]
{\small
  \[
\begin{array}{c}
\infer{\cdot \vdash \mathsf{wf} }{}    
\quad\quad
\infer{\Phi, \alpha : H \vdash \mathsf{wf} }{\Phi \vdash \mathsf{wf}}    
\quad\quad
\infer{\Phi, \kappa : H \vdash \mathsf{wf} }{\Phi \vdash \mathsf{wf}}    
\quad\quad
\infer{\Phi, e : H \vdash \mathsf{wf} }{\Phi \vdash H \Downarrow e}    
\end{array}  
\]}
\end{definition}

%
%
%
%
As an example, let us consider resolving the candidate lemma \ensuremath{\Conid{Eq}\;\Varid{x}\Rightarrow \Conid{Eq}\;(\Conid{My}\;\Conid{HPTree}\;\Varid{x})} against
the axiom environment $\Phi_{\ensuremath{\Conid{HPTree}}}$. This yields the following derivation, where
$\Phi_1 = \Phi_{\ensuremath{\Conid{HPTree}}}, (\alpha : \ensuremath{\Conid{Eq}\;\Varid{x}\Rightarrow \Conid{Eq}\;(\Conid{Mu}\;\Conid{HPTree}\;\Varid{x})})$ and $\Phi_2 = \Phi_1, (\alpha_1 : \ensuremath{\Conid{Eq}\;\Conid{C}})$:
%

{\small
\begin{center}
  \infer{\Phi_{\ensuremath{\Conid{HPTree}}}\vdash \ensuremath{\Conid{Eq}\;\Varid{x}\Rightarrow \Conid{Eq}\;(\Conid{Mu}\;\Conid{HPTree}\;\Varid{x})} \Downarrow \mu \alpha . \lambda \alpha_1 . \kappa_{\ensuremath{\Conid{Mu}}}\ (\kappa_{\ensuremath{\Conid{HPTree}}}\ \alpha_1\ (\alpha\
  (\kappa_{\ensuremath{\Conid{Pair}}}\ \alpha_1\ \alpha_1)))}
   {\infer{\Phi_1 \vdash \ensuremath{\Conid{Eq}\;\Varid{x}\Rightarrow \Conid{Eq}\;(\Conid{Mu}\;\Conid{HPTree}\;\Varid{x})} \Downarrow \lambda \alpha_1 . \kappa_{\ensuremath{\Conid{Mu}}}
\ (\kappa_{\ensuremath{\Conid{HPTree}}}\ \alpha_1\ (\alpha\
  (\kappa_{\ensuremath{\Conid{Pair}}}\ \alpha_1\ \alpha_1)))}
   {\infer{(\Phi_2 = \Phi_1, \alpha_1 : \ensuremath{\Conid{Eq}\;\Conid{C}}) \vdash \ensuremath{\Conid{Eq}\;(\Conid{Mu}\;\Conid{HPTree}\;\Conid{C})} \Downarrow \kappa_{\ensuremath{\Conid{Mu}}}\ (\kappa_{\ensuremath{\Conid{HPTree}}}\ \alpha_1\ (\alpha\
  (\kappa_{\ensuremath{\Conid{Pair}}}\ \alpha_1\ \alpha_1)))}
{\infer{\Phi_2 \vdash \ensuremath{\Conid{Eq}\;(\Conid{HPTree}\;(\Conid{Mu}\;\Conid{HPTree}\;\Conid{C}))} \Downarrow \kappa_{\ensuremath{\Conid{HPTree}}}\ \alpha_1\ (\alpha\
  (\kappa_{\ensuremath{\Conid{Pair}}}\ \alpha_1\ \alpha_1))}{
   \infer{\Phi_2 \vdash \ensuremath{\Conid{Eq}\;\Conid{C}} \Downarrow \alpha_1}{} & \infer{\Phi_2 \vdash \ensuremath{\Conid{Eq}\;(\Conid{HPTree}\;(\Conid{C},\Conid{C}))} \Downarrow \alpha\
  (\kappa_{\ensuremath{\Conid{Pair}}}\ \alpha_1\ \alpha_1)}{
   \infer{\Phi_2 \vdash \ensuremath{\Conid{Eq}\;(\Conid{C},\Conid{C})} \Downarrow (\kappa_{\ensuremath{\Conid{Pair}}}\ \alpha_1\ \alpha_1)}
{\infer{\Phi_2 \vdash \ensuremath{\Conid{Eq}\;\Conid{C}} \Downarrow \alpha_1}{} & \infer{\Phi_2 \vdash \ensuremath{\Conid{Eq}\;\Conid{C}} \Downarrow \alpha_1}{}}}}} }}
\end{center}
}

\noindent Once we prove \ensuremath{\Conid{Eq}\;\Varid{x}\Rightarrow \Conid{Eq}\;(\Conid{Mu}\;\Conid{HPTree}\;\Varid{x})} from $\Phi_{\ensuremath{\Conid{HPTree}}}$ by corecursive resolution, we can add it to the
axiom environment and use it to prove the ground query \ensuremath{\Conid{Eq}\;(\Conid{Mu}\;\Conid{HPTree}\;\Conid{Int})}. Let $\Phi' = \Phi_{\ensuremath{\Conid{HPTree}}}, (\mu \alpha . \lambda \alpha_1 . \kappa_1\ (\kappa_2\ \alpha_1\  (\alpha\ (\kappa_3\ \alpha_1\ \alpha_1))) : \ensuremath{\Conid{Eq}\;\Varid{x}\Rightarrow \Conid{Eq}\;(\Conid{Mu}\;\Conid{HPTree}\;\Varid{x})})$. We have the following derivation. 

{\small
\[
\begin{array}{c}
  \infer{\Phi' \vdash \ensuremath{\Conid{Eq}\;(\Conid{Mu}\;\Conid{HPTree}\;\Conid{Int})} \Downarrow (\mu \alpha . \lambda \alpha_1 . \kappa_{\ensuremath{\Conid{Mu}}}\ (\kappa_{HPTree}\ \alpha_1\  (\alpha\ (\kappa_{\ensuremath{\Conid{Pair}}}\ \alpha_1\ \alpha_1)))) \
    \kappa_{\ensuremath{\Conid{Int}}}} {\Phi' \vdash \ensuremath{\Conid{Eq}\;\Conid{Int}} \Downarrow \kappa_{\ensuremath{\Conid{Int}}}}
\end{array}
\]
}

\section{Operational Semantics of Corecursive Evidence}
\label{ob-equiv}

The purpose of this section is to give operational
semantics to corecursive resolution, and in particular, we are interested in giving 
operational interpretation to the
corecursive evidence constructed as a result of applying corecursive resolution.
In type class applications, for example, the evidence constructed for a query will be run as a program.
It is therefore important to establish the exact relationship
between the non-terminating resolution as a process and the proof-term that we obtain via corecursive resolution. 
We prove that corecursive evidence indeed faithfully
captures the otherwise infinite resolution process of Section~\ref{preliminary}.


In general, we know that if $\Phi \vdash A \to^* \mathcal{C}[\sigma A]$, then we can observe the following looping infinite reduction trace:
\[ \Phi \vdash A \to^* \mathcal{C}[\sigma A] \to^* \mathcal{C}[\sigma \mathcal{C}[\sigma^2 A]] \to^* \mathcal{C}[\sigma \mathcal{C}[\sigma^2 \mathcal{C}[\sigma^3 A]]] \to ... \]
Each iteration of the loop gives rise to repeatedly applying
substitution $\sigma$ to the reduction context $\mathcal{C}$. 

In principle, this mixed term context $\mathcal{C}$ may contain an atomic
formula $B$ that itself is normalizing, where $\sigma B$ spawns another loop. 
Clearly this is a complicating factor. For instance, a
loop can spawn off additional loops in each iteration. Alternatively, a loop
can have multiple iteration points such as 
$\Phi \vdash A \to^* \mathcal{C}[\sigma_1 A, \sigma_2 A, ..., \sigma_n
A]$.\footnote{Note that we abuse notation here to denote contexts with multiple holes.
Also we abbreviate identical instantiation of $\mathcal{C}[D, \ldots, D]$ those multiple holes  to $\mathcal{C}[D]$.} 
These complicating factors are beyond the scope of this section. We focus
only on \textit{simple loops}. These are loops with a single iteration point
that does not spawn additional loops.

We use $|\mathcal{C}|$ to denote the set of atomic formulas in the context
$\mathcal{C}$. If all atomic formulas $D \in |\mathcal{C}|$ are irreducible
with respect to $\Phi$, then we call $\mathcal{C}$ a \textit{normal context}. 
\begin{definition}[Simple Loop]
Let $\Phi \vdash B \to^* \mathcal{C}[\sigma B]$, where $\mathcal{C}$ is
normal. If for all $D \in |\mathcal{C}|$, we have that $\Phi \vdash \sigma D \to^*
\mathcal{C}'[D]$ with $| \mathcal{C}' | = \emptyset$, then we call $\Phi
\vdash B \to^* \mathcal{C}[\sigma B]$ a simple loop.  
 \end{definition}
In the above definition, the normality of $\mathcal{C}$ ensures that the loop
has a single iteration point. Likewise the condition $\Phi \vdash \sigma D \to^*
\mathcal{C}'[D]$, which implies that $\Phi \vdash \sigma^n D \to^*
\mathcal{C}'^n[D]$, guarantees that each iteration of the loop
spawns no further loops.

 \begin{definition}[Observational Point]
\label{observe}
   Let $\Phi \vdash B \to^* \mathcal{C}'[\sigma B]$ be a simple loop and $\Phi \vdash B \to^* q$.   We call $q$ an observational point if it is of the form $\mathcal{C}[\delta B]$. We use $\mathcal{O}(B)_{\Phi}$ to denote the set of observational points in the simple loop. 
 \end{definition}

For example, we have the following infinite resolution trace generated by the simple loop (with the subterms of observational points underlined). 
 {\scriptsize
 \begin{center}

   $\Phi_{\ensuremath{\Conid{HPTree}}} \vdash \underline{{Eq}\ ({Mu}\ {HPTree} \ x)} \to \kappa_{\ensuremath{\Conid{Mu}}}\ ({Eq}\ ({HPTree}\ ({Mu}\ {HPTree})\ x)) \to \kappa_{\ensuremath{\Conid{Mu}}}\ (\kappa_{\ensuremath{\Conid{HPTree}}}\ ({Eq}\ x) \ \underline{({Eq}\ ({Mu}\ {HPTree}\ (x,x)))}) \to 
 \kappa_{\ensuremath{\Conid{Mu}}}\ (\kappa_{\ensuremath{\Conid{HPTree}}}\ ({Eq}\ x) \ (\kappa_{\ensuremath{\Conid{Mu}}} \ (Eq \ (HPTree \ (Mu \ HPTree)\ (x,x))))) \to
  \kappa_{\ensuremath{\Conid{Mu}}} (\kappa_{\ensuremath{\Conid{HPTree}}} ({Eq}\ x) (\kappa_{\ensuremath{\Conid{Mu}}} (\kappa_{\ensuremath{\Conid{HPTree}}}  (Eq  (x,x)) (Eq (Mu\ HPTree \ ((x,x),(x,x))))))) \to \kappa_{\ensuremath{\Conid{Mu}}} (\kappa_{\ensuremath{\Conid{HPTree}}} ({Eq}\ x) (\kappa_{\ensuremath{\Conid{Mu}}} (\kappa_{\ensuremath{\Conid{HPTree}}}  (\kappa_{\ensuremath{\Conid{Pair}}}  (Eq\ x) (Eq\ x))) \underline{\ensuremath{(\Conid{Eq}\;(\Conid{Mu}\;\Conid{HPTree}\;((\Varid{x},\Varid{x}),(\Varid{x},\Varid{x}))))}})))$ $\to ... $
 \end{center}
 }
\noindent In this case, we have $\sigma = [\ensuremath{(\Varid{x},\Varid{x})}/\ensuremath{\Varid{x}}]$ and $\Phi \vdash \sigma ({Eq}\ x) \to \kappa_{\ensuremath{\Conid{Pair}}}\ ({Eq}\ x)\ ({Eq}\ x)$.


The corecursive evidence encapsulates an infinite derivation in a finite fixpoint
expression. We can recover the infinite resolution by reducing the
corecursive expression.
To define small-step \textit{evidence reduction}, we first extend
 mixed terms to cope with richer
corecursive evidence.
\begin{definition}
Mixed term $q \  ::=  \ A \ |\ \kappa \ | \ q \ q'\ | \ \textcolor{blue}{\alpha}\ |\ \textcolor{blue}{\lambda \alpha.q} \ | \ \textcolor{blue}{\mu \alpha.q}$
\end{definition}

Now we define the small-step evidence reduction relation $q \leadsto q'$.
\begin{definition}[Small Step Evidence Reduction]
  \label{label}
  \fbox{$q \leadsto q'$}
{\small
\[\begin{array}{c}

\infer{{\mathcal{C}[\mu \alpha . q]} \leadsto_\mu \mathcal{C}[[\mu \alpha . q/\alpha]q]}{} 
\quad \quad
\infer{\mathcal{C}[(\lambda \alpha . q) \ q'] \leadsto_\beta \mathcal{C}[[q'/\alpha] q] }{}

  \end{array}
\]}

\end{definition}
Note that for simplicity we still use the mixed term context $\mathcal{C}$ as defined in
Section \ref{preliminary}, but we only allow the reduction of an outermost redex, i.e.,
a redex that is not a subterm of some other redex. In other words, 
reduction unfolds the evidence term strictly downwards from the
root, this follows closely the way evidence is constructed during
resolution.

We call the states where we perform a $\mu$-transition \emph{corecursive points}.
Note that $\mu$-transitions unfold
 a corecursive definition.
These
correspond closely to the observational points in resolution.

\begin{definition}[Corecursive Point]
Let $q' \leadsto^* q$. We call $q$ a corecursive point if it is of the form $\mathcal{C}[(\mu \alpha.e) \ q_1 ...\ q_n]$. We use $\mathcal{S}(q')$ to denote the set of corecursive points in $q' \leadsto^* q$.  
 \end{definition}

Let $e \equiv \mu \alpha . \lambda \alpha_1 . \kappa_{\ensuremath{\Conid{Mu}}}\ (\kappa_{\ensuremath{\Conid{HPTree}}}\ \alpha_1\  (\alpha\ (\kappa_{\ensuremath{\Conid{Pair}}}\ \alpha_1\ \alpha_1)))$. We have the following evidence reduction trace (with the subterms of corecursive points underlined):

{\scriptsize
  \begin{center}

      $\underline{e \ ({Eq}\ x)} \leadsto_\mu (\lambda \alpha_1
      . \kappa_{\ensuremath{\Conid{Mu}}}\ (\kappa_{\ensuremath{\Conid{HPTree}}}\ \alpha_1 \ (e \
      (\kappa_{\ensuremath{\Conid{Pair}}} \ \alpha_1 \ \alpha_1))))\ ({Eq}\ x)
      \leadsto_\beta \kappa_{\ensuremath{\Conid{Mu}}}\ (\kappa_{\ensuremath{\Conid{HPTree}}}\ ({Eq}\ x) \
      \underline{(e \ (\kappa_{\ensuremath{\Conid{Pair}}} \ ({Eq}\ x) \ ({Eq}\ x)))})
      \leadsto_\mu \kappa_{\ensuremath{\Conid{Mu}}}\ (\kappa_{\ensuremath{\Conid{HPTree}}}\ ({Eq}\ x) \
      ((\lambda \alpha_1 . \kappa_{\ensuremath{\Conid{Mu}}}\ (\kappa_{\ensuremath{\Conid{HPTree}}}\ \alpha_1
      \ (e \ (\kappa_{\ensuremath{\Conid{Pair}}} \ \alpha_1 \ \alpha_1)))) \
      (\kappa_{\ensuremath{\Conid{Pair}}} \ ({Eq}\ x) \ ({Eq}\ x)))) \leadsto_\beta
      \kappa_{\ensuremath{\Conid{Mu}}} (\kappa_{\ensuremath{\Conid{HPTree}}} ({Eq}\ x) (\kappa_{\ensuremath{\Conid{Mu}}}
      (\kappa_{\ensuremath{\Conid{HPTree}}} (\kappa_{\ensuremath{\Conid{Pair}}} ({Eq}\ x) ({Eq}\ x))\underline{(e (\kappa_{\ensuremath{\Conid{Pair}}} (\kappa_{\ensuremath{\Conid{Pair}}} ({Eq}\ x) ({Eq}\
        x)) (\kappa_{\ensuremath{\Conid{Pair}}} ({Eq}\ x) ({Eq}\ x)) ))})))$
      $\leadsto_\mu ...$
  \end{center}
}

Observe that the mixed term contexts of the observational points and the
corecursive points in the above traces coincide. This allows us to show
observational equivalence of resolution and evidence reduction without
explicitly introducing actual infinite evidence.

The following theorem shows that if resolution gives rise to a simple loop,
then we can obtain a corecursive evidence $e$ (Theorem \ref{equiv} (1)) such that
the infinite resolution trace is observational equivalent to $e$'s evidence
reduction trace (Theorem \ref{equiv} (2)). 

\begin{theorem}[Observational Equivalence]
\label{equiv}
Let $\Phi \vdash B \to^* \mathcal{C}[\sigma B]$ be a simple loop and $|\mathcal{C}| = \{D_1, ..., D_n\}$.
Then:

\noindent \textit{1.} We have  $\Phi \vdash D_1, ..., D_n \Rightarrow B \Downarrow \mu \alpha . \lambda \alpha_1 .... \lambda \alpha_n . e$ for some $e$. 

\noindent \textit{2.} $\mathcal{C}[\delta B] \in \mathcal{O}(B)_\Phi$ iff $\mathcal{C}[(\mu \alpha . \lambda \underline{\alpha} . e) \ \underline{q}] \in \mathcal{S}((\mu \alpha . \lambda \underline{\alpha} . e) \ \underline{D})$.
\end{theorem}

The proof can be found in the extended version.

 



\section{Related Work}
\label{relwork}
\emph{Calculus of Coinductive Constructions.} Interactive theorem prover Coq pioneered implementation of the \textit{guarded coinduction principle} (\cite{coquand1994infinite,gimenez1996}).  
The Coq termination checker may prevent some nested uses of coinduction, 
e.g. a proof term such as $(\mu \alpha . \lambda x . \kappa_0\ (\kappa_1\ x\ (\alpha\ (\alpha\ x)))) \ \kappa_2$ is not accepted by Coq, while from the outermost reduction point of view, this proof term is productive. 

\emph{Loop detection in term rewriting.} Distinctions between cycle, loop and non-looping
 has long been established in term rewriting research (\cite{dershowitz1987,zantema1996non}). For us, detecting loop is the first step of invariant analysis, but we also want to extract corecursive evidence such that it captures the infinite reduction trace.

\emph{Restricted datatypes.} Hughes (Section 4 \cite{hughes1999restricted})
observed the cyclic nature of the instance declarations \ensuremath{\mathbf{instance}\;\Conid{Sat}\;(\Conid{EqD}\;\Varid{a})\Rightarrow \Conid{Eq}\;\Varid{a}} and
\ensuremath{\mathbf{instance}\;\Conid{Eq}\;\Varid{a}\Rightarrow \Conid{Sat}\;(\Conid{EqD}\;\Varid{a})}. He proposed to treat the looping context
reduction as failure, in which case the compiler would need to search for an alternative
reduction.

\emph{Scrap Your Boilerplate with Class.} The cycle
detection method \cite{Lammel:2005} was proposed to
generate corecursive evidence for a restricted
class of non-terminating resolution. It is supported by the current Glasgow Haskell Compiler.

\emph{Derivable type classes.} Hinze and Peyton Jones \cite{hinze2001}
wanted to use an instance of the form
\noindent \ensuremath{\mathbf{instance}\;(\Conid{Binary}\;\Varid{a},\Conid{Binary}\;(\Varid{f}\;(\Conid{GRose}\;\Varid{f}\;\Varid{a})))\Rightarrow \Conid{Binary}\;(\Conid{GRose}\;\Varid{f}\;\Varid{a})}, but discovered that it causes resolution to diverge.
They  suggested the following as a replacement:
\noindent \ensuremath{\mathbf{instance}\;(\Conid{Binary}\;\Varid{a},\forall \Varid{b}\hsforall \hsdot{\mathrel{\,\cdot\,}}{\;.\;}\Conid{Binary}\;\Varid{b}\Rightarrow \Conid{Binary}\;\Varid{f}\;\Varid{b})\Rightarrow \Conid{Binary}\;(\Conid{GRose}\;\Varid{f}\;\Varid{a})}. 
Unfortunately, Haskell does not support instances with \textit{polymorphic higher-order instance contexts}.
Nevertheless, allowing such implication constraints  
would greatly increase the expressitivity of corecursive resolution. In the terminology of 
our paper, it amounts to extending Horn formulas to intuitionistic formulas. 
Working with intuitionistic formulas would require a certain amount of searching, 
as the non-overlapping condition for Horn formulas is not enough to ensure uniqueness of 
the evidence. For example, consider the following axioms: 

  \begin{center}
    $\kappa_1 : (A \Rightarrow B\ x) \Rightarrow D\ (S\ x)$

    $\kappa_2 : A, D\ x \Rightarrow B\ (S\ x)$
    
    $\kappa_3 :\ \Rightarrow D\ Z$
  \end{center}

  

\noindent  We have two distinct proof terms for $D\ (S\ (S\ (S\ (S\ Z)))))$:

\begin{center}
  $\kappa_1\ (\lambda \alpha_1 . \kappa_2\ \alpha_1\ (\kappa_1\ (\lambda \alpha_2
  . \kappa_2\ \underline{\remph{\alpha_1}}\ \kappa_3)))$

  $\kappa_1\ (\lambda \alpha_1 . \kappa_2\ \alpha_1\ (\kappa_1\ (\lambda \alpha_2
  . \kappa_2\ \underline{\remph{\alpha_2}}\ \kappa_3)))$
\end{center}
\noindent This is 
 undesirable from the perspective of generating evidence for type class.  

\emph{Instance declarations and (Horn Clause) logic programs}. The process of simplifying type class
constraints is formally described as the notion of \textit{context reduction}
by Peyton Jones et. al.~\cite{Jones97}. In Section 3.2 of the same paper also
describes the form and the requirements of instance declarations. Type class
evidence and their connection to type system are studied in Mark Jones's
thesis~\cite[Chapter 4.2]{jones2003qualified}. Context reduction,
instance declaration and their connection to proof relevant resolution are also
discussed under the name of \textit{LP-TM} (logic programming with
term-matching) in Fu and Komendantskaya~\cite[Section 4.1]{FuK15}.


\section{Conclusion and Future Work}\label{fwork}

We have introduced a novel approach to  non-terminating resolution.
Firstly, we have shown that the popular cycle detection methods employed for logic programming or type class resolution
can be understood via more general coinductive proof principles (\cite{coquand1994infinite,gimenez1996}).
Secondly, we have shown that 
 resolution can be enriched with rules that capture the intuition of richer coinductive hypothesis formation.
 This extension allows to provide corecursive evidence to some derivations that could not be handled by previous methods.
 Moreover, corecursive resolution is formulated in a proof-relevant way, i.e. proof-evidence construction is an essential part
 of  corecursive resolution. This makes it easier to integrate it directly into type class inference.

We have implemented the techniques of Sections \ref{inv} and \ref{corec}, and have incorporated
them as part of the evidence construction process for a simple language that 
admits previously non-terminating examples.\footnote{See the extended version for more examples and information about the implementation. Extended version is available from authors' homepages.} 

\textbf{Future Work} In general, the interactions between different loops
can be complicated. Consider $\Phi_{\ensuremath{\Conid{Pair}}}$ with the following
declarations (denoted by $\Phi_M$):
\begin{hscode}\SaveRestoreHook
\column{B}{@{}>{\hspre}l<{\hspost}@{}}%
\column{3}{@{}>{\hspre}l<{\hspost}@{}}%
\column{E}{@{}>{\hspre}l<{\hspost}@{}}%
\>[3]{}\kappa_\Conid{M}\mathbin{:}\Conid{Eq}\;(\Varid{h}_{1}\;(\Conid{M}\;\Varid{h}_{1}\;\Varid{h}_{2})\;(\Conid{M}\;\Varid{h}_{2}\;\Varid{h}_{1})\;\Varid{a})\Rightarrow \Conid{Eq}\;(\Conid{M}\;\Varid{h}_{1}\;\Varid{h}_{2}\;\Varid{a}){}\<[E]%
\\
\>[3]{}\kappa_\Conid{H}\mathbin{:}(\Conid{Eq}\;\Varid{a},\Conid{Eq}\;((\Varid{f}_{1}\;\Varid{a}),(\Varid{f}_{2}\;\Varid{a})))\Rightarrow \Conid{Eq}\;(\Conid{H}\;\Varid{f}_{1}\;\Varid{f}_{2}\;\Varid{a}){}\<[E]%
\\
\>[3]{}\kappa_\Conid{G}\mathbin{:}\Conid{Eq}\;((\Varid{g}\;\Varid{a}),(\Varid{f}\;(\Varid{g}\;\Varid{a})))\Rightarrow \Conid{Eq}\;(\Conid{G}\;\Varid{f}\;\Varid{g}\;\Varid{a}){}\<[E]%
\ColumnHook
\end{hscode}\resethooks
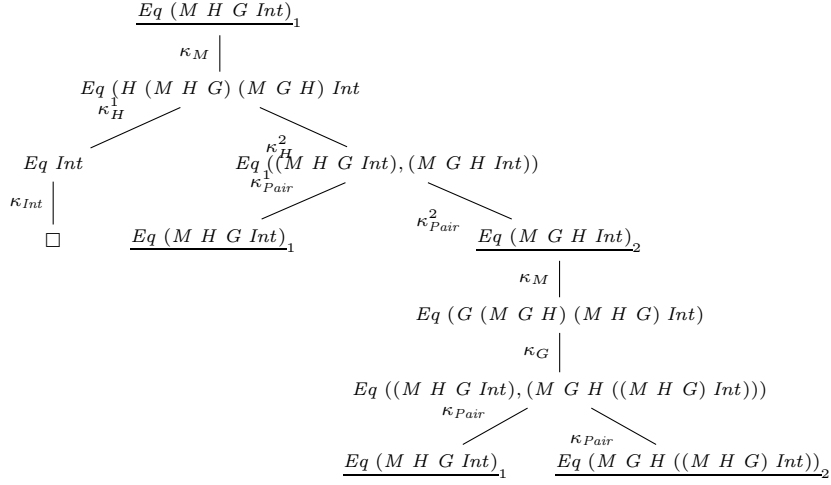
\begin{figure}[tbp]
  \makebox[\textwidth]{
\begin{tikzpicture}[every tree node/.style={},
   level distance=1.00cm,sibling distance=.40cm, 
   edge from parent path={(\tikzparentnode) -- (\tikzchildnode)}, font=\scriptsize]
\Tree [.\node{$\underline{\ensuremath{\Conid{Eq}\;(\Conid{M}\;\Conid{H}\;\Conid{G}\;\Conid{Int})}}_1$};
        \edge node[auto=right]{$\kappa_{M}$};      
         [.\node{\ensuremath{\Conid{Eq}\;(\Conid{H}\;(\Conid{M}\;\Conid{H}\;\Conid{G})\;(\Conid{M}\;\Conid{G}\;\Conid{H})\;\Conid{Int}}};  
        \edge node[auto=right]{$\kappa_{H}^1$};      
         [.\node{\ensuremath{\Conid{Eq}\;\Conid{Int}}};  
           \edge node[auto=right]{$\kappa_{\ensuremath{\Conid{Int}}}$};      
         [.\node{$\Box$};  
         ]
         ]
         \edge node[auto=right]{$\kappa_H^2$};      
         [.\node{\ensuremath{\Conid{Eq}\;((\Conid{M}\;\Conid{H}\;\Conid{G}\;\Conid{Int}),(\Conid{M}\;\Conid{G}\;\Conid{H}\;\Conid{Int}))}};  
           \edge node[auto=right]{$\kappa_{\ensuremath{\Conid{Pair}}}^1$};      
         [.\node{$\underline{\ensuremath{\Conid{Eq}\;(\Conid{M}\;\Conid{H}\;\Conid{G}\;\Conid{Int})}}_1$};  
         ]
         \edge node[auto=right]{$\kappa_{\ensuremath{\Conid{Pair}}}^2$};      
         [.\node{$\underline{\ensuremath{\Conid{Eq}\;(\Conid{M}\;\Conid{G}\;\Conid{H}\;\Conid{Int})}}_2$};  
           \edge node[auto=right]{$\kappa_M$};      
         [.\node{\ensuremath{\Conid{Eq}\;(\Conid{G}\;(\Conid{M}\;\Conid{G}\;\Conid{H})\;(\Conid{M}\;\Conid{H}\;\Conid{G})\;\Conid{Int})}};  
         \edge node[auto=right]{$\kappa_G$};      
         [.\node{\ensuremath{\Conid{Eq}\;((\Conid{M}\;\Conid{H}\;\Conid{G}\;\Conid{Int}),(\Conid{M}\;\Conid{G}\;\Conid{H}\;((\Conid{M}\;\Conid{H}\;\Conid{G})\;\Conid{Int})))}};  
           \edge node[auto=right]{$\kappa_{\ensuremath{\Conid{Pair}}}$};      
         [.\node{$\underline{\ensuremath{\Conid{Eq}\;(\Conid{M}\;\Conid{H}\;\Conid{G}\;\Conid{Int})}}_1$};  
         ]
         \edge node[auto=right]{$\kappa_{\ensuremath{\Conid{Pair}}}$};      
         [.\node{$\underline{\ensuremath{\Conid{Eq}\;(\Conid{M}\;\Conid{G}\;\Conid{H}\;((\Conid{M}\;\Conid{H}\;\Conid{G})\;\Conid{Int}))}}_2$};  
         ]
         ]
         ]
         ]
         ]
         ]
         ]
\end{tikzpicture}
}
	\caption{A Partial Resolution tree for $\Phi_M \vdash \ensuremath{\Conid{Eq}\;(\Conid{M}\;\Conid{H}\;\Conid{G}\;\Conid{Int})} \Uparrow$}
	\label{fig:ex4}
\end{figure}
\noindent A partial resolution tree generated by the query \ensuremath{\Conid{Eq}\;(\Conid{M}\;\Conid{H}\;\Conid{G}\;\Conid{Int})} is
described in Figure \ref{fig:ex4}. In this case the cycle (underlined with the
index $1$) is mutually nested with a loop (underlined with index $2$). Our method in
Section \ref{inv} is not able to generate any candidate lemmas. Yet there are two
candidate lemmas for this case (with the proof of $e_2$
refer to $e_1$): 

\begin{hscode}\SaveRestoreHook
\column{B}{@{}>{\hspre}l<{\hspost}@{}}%
\column{3}{@{}>{\hspre}l<{\hspost}@{}}%
\column{E}{@{}>{\hspre}l<{\hspost}@{}}%
\>[3]{}\Varid{e}_{1}\mathbin{:}(\Conid{Eq}\;\Varid{x},\Conid{Eq}\;(\Conid{M}\;\Conid{G}\;\Conid{H}\;\Varid{x}))\Rightarrow \Conid{Eq}\;(\Conid{M}\;\Conid{H}\;\Conid{G}\;\Varid{x}){}\<[E]%
\\
\>[3]{}\Varid{e}_{2}\mathbin{:}\Conid{Eq}\;\Varid{x}\Rightarrow \Conid{Eq}\;(\Conid{M}\;\Conid{G}\;\Conid{H}\;\Varid{x}){}\<[E]%
\ColumnHook
\end{hscode}\resethooks
\noindent 
We would
like to improve our heuristics to allow generating multiple candidate lemmas, where
their corecursive evidences mutually refer to each other. 

There are situations where resolution is non-terminating but does not form any loop such
as $\Phi \vdash A \to^* \mathcal{C}[\sigma A]$. Consider the following program $\Phi_D$: 
\begin{hscode}\SaveRestoreHook
\column{B}{@{}>{\hspre}l<{\hspost}@{}}%
\column{3}{@{}>{\hspre}l<{\hspost}@{}}%
\column{E}{@{}>{\hspre}l<{\hspost}@{}}%
\>[3]{}\kappa_\mathrm{1}\mathbin{:}\Conid{D}\;\Varid{n}\;(\Conid{S}\;\Varid{m})\Rightarrow \Conid{D}\;(\Conid{S}\;\Varid{n})\;\Varid{m}{}\<[E]%
\\
\>[3]{}\kappa_\mathrm{2}\mathbin{:}\Conid{D}\;(\Conid{S}\;\Varid{m})\;\Conid{Z}\Rightarrow \Conid{D}\;\Conid{Z}\;\Varid{m}{}\<[E]%
\ColumnHook
\end{hscode}\resethooks
\noindent For query \ensuremath{\Conid{D}\;\Conid{Z}\;\Conid{Z}}, the resolution diverges without forming any loop. We would have
to introduce extra equality axioms in order to obtain finite corecursive evidence.\footnote{See the extended version for a solution in Haskell using type family and more discussion.} We would like
to investigate the ramifications of incorporating equality axioms in the corecursive resolution in the future. 

We plan to extend the observational equivalence result of Section
\ref{ob-equiv} to cope with more general notions of loop and extend our approach to 
allow intuitionistic formulas as candidate lemmas. 


Another avenue for future work is a formal proof that  the calculus of Definition~\ref{def:cresolution} is sound relative to the the greatest Herbrand models~\cite{Llo87},
 and therefore reflects the broader notion of the coinductive entailment for Horn clause logic as discussed in the introduction.

\subsection*{Acknowledgements} We want to thank Patricia Johann and the FLOPS reviewers for their helpful 
comments, Franti\v sek Farka for many discussions. Part of this work was funded
by the Flemish Fund for Scientific Research.

\bibliographystyle{plain}
\bibliography{paper}

\newpage
\appendix


\section{Proof of Theorem \ref{equiv}}\label{app:proof}






 
\begin{theorem}
Let $\Phi \vdash B \to^* \mathcal{C}[D_1, ..., D_n, \sigma B]$
with $|\mathcal{C}| = \emptyset$ and $D_i$ are normal for all $i$. Suppose $\Phi \vdash \sigma D_i \to^* \mathcal{C}_i[D_i]$, where $| \mathcal{C}_i | = \emptyset$ for any $D_i$. We have the following: 

\begin{enumerate}
\item $\Phi \vdash D_1, ..., D_n \Rightarrow B \Downarrow \mu \alpha . \lambda \alpha_1 .... \lambda \alpha_n . e$.
\item Let $e' \equiv \mu \alpha . \lambda \alpha_1 .... \lambda \alpha_n . e$. We have: 

$\mathcal{C}'[\sigma^m B] \in \mathcal{O}(B)_{\Phi}$ iff $\mathcal{C}'[e' \ \mathcal{C}_1^m[D_1] ...\ \mathcal{C}_n^m[D_n]] \in \mathcal{S}(e' \ \underline{D})$. 
\end{enumerate}
\end{theorem}
\begin{proof}

  \begin{enumerate}
  \item We have the following finite derivation. 

    \begin{center}
\[
      \begin{array}{c}
        \infer{\Phi \vdash \underline{D} \Rightarrow B \Downarrow
          \mu \alpha . \lambda \alpha_1 .... \lambda \alpha_n
            . e} { \infer{\Phi, \alpha : \underline{D}
            \Rightarrow B \vdash \underline{D} \Rightarrow B
            \Downarrow \lambda \alpha_1 .... \lambda \alpha_n
            . e}{\infer{\Phi, \alpha : \underline{D} \Rightarrow B,
              \underline{\alpha} : \underline{D} \vdash B \Downarrow
              e}{\infer{...}{}}}}
      \end{array} \]
    \end{center}

\noindent By Theorem \ref{func:eq}, we just need to reduce $B$ to a proof term using
    the rules $\Psi_1 = \Phi, \alpha : \underline{D} \Rightarrow B, \underline{\alpha} : \underline{D}$. 
We have the following reduction:

    \begin{center}
      $ \Psi_1 \vdash B \to^* \mathcal{C}[D_1, ..., D_n, (\sigma B)]
      \to^* \mathcal{C}[\alpha_1 ... \alpha_n, (\sigma B)] \to
      \mathcal{C}[\alpha_1, ... ,\alpha_n, (\alpha\ (\sigma\ D_1)
      ... (\sigma\ D_n)) ] \to^* \mathcal{C}[ \alpha_1, ... \alpha_n,
      (\alpha\ \mathcal{C}_1[D_1] ... \mathcal{C}_n[D_n])] \to
      \mathcal{C}[\alpha_1 ... \alpha_n (\alpha\
      \mathcal{C}_1[\alpha_1] ...\ \mathcal{C}_n[\alpha_n])]$
    \end{center}
    
    Thus we have the corecursive evidence 
    $\mu \alpha .\lambda \alpha_1 . .. \alpha_n
    . e$ for $ D_1, ...,
    D_n \Rightarrow B$, and $e \equiv \mathcal{C}[ \alpha_1, ... \alpha_n, (\alpha\
    \mathcal{C}_1[\alpha_1]\ ...\ \mathcal{C}_n[\alpha_n])]$.
  \item Using the same notation in (1), let $e' \equiv \mu \alpha .\lambda \alpha_1 . .. \alpha_n
    . e$, we can observe following equivalence reduction traces:
 
  \begin{center}
  $\Phi \vdash B \to^* \mathcal{C}[D_1, ..., D_n, (\sigma B)] \to^* \mathcal{C}[D_1, ..., D_n, \mathcal{C}[\mathcal{C}_1[D_1], ..., \mathcal{C}_n[D_n], (\sigma^2 B)]] \to ...$
   
\

    $\ (e'\ D_1 ...\ D_n) \leadsto^* \mathcal{C}[ D_1, ... D_n, (e'\ \mathcal{C}_1[D_1]\ ...\ \mathcal{C}_n[D_n])] \leadsto^* \mathcal{C}[ D_1, ..., D_n, \mathcal{C}[ \mathcal{C}_1[D_1], ..., \mathcal{C}_n [D_n], (e'\ (\mathcal{C}_1[\mathcal{C}_1[D_1]])\ ...\ (\mathcal{C}_n[\mathcal{C}_n[D_n]]))]] \leadsto^* ...$
  \end{center}
We proceed by induction on $m$. When $m = 0$, it is obvious. Let $m = k + 1$. By IH, 
we know that $\mathcal{C}'[\sigma^k B] \in \mathcal{O}(B)_{\Phi}$ iff $\mathcal{C}'[e' \ \mathcal{C}_1^k[D_1] ...\ \mathcal{C}_n^k[D_n]] \in \mathcal{S}(e'\ \underline{D})$. We have $\Phi \vdash \mathcal{C}'[\sigma^k B] \to^* \mathcal{C}'[\mathcal{C}[\sigma^k D_1, ..., \sigma^k D_n, (\sigma^{k+1} B)]] \to^*$

$\mathcal{C}'[\mathcal{C}[\mathcal{C}_n^k[D_1], ..., \mathcal{C}_n^k[D_n], (\sigma^{k+1} B)]]$.

On the other hand, $\mathcal{C}'[e' \ \mathcal{C}_1^k[D_1] ...\ \mathcal{C}_n^k[D_n]] \leadsto^*$

$\mathcal{C}'[\mathcal{C}[\mathcal{C}_n^k[D_1], ..., \mathcal{C}_n^k[D_n], (e' \ \mathcal{C}_1^{k+1}[D_1] ...\ \mathcal{C}_n^{k+1}[D_n])]]$. Thus we have the observational equivalence. 
  \end{enumerate}
\end{proof}

\section{Weak Head Normalization of Corecursive Evidence}
\label{ap:guard}
\begin{definition}
\label{syntax}
\[
\begin{array}{l@{\hspace{5mm}}l@{\hspace{5mm}}c@{\hspace{5mm}}l}
    \text{Formula}         &  F, G & ::= & A ~\mid~ \forall x . F ~\mid~ \ F \Rightarrow F' \\
    \text{Evidence/Proofs} &  e    & ::= & \kappa ~\mid~ \alpha ~\mid~ \lambda \alpha . e ~\mid~ e \ e' ~\mid~
\mu \alpha . e  \\
    \text{Contexts/Axioms}  &  \Phi & ::= & \cdot ~\mid~ e : F, \Phi
\end{array}
\]
\end{definition}

We write $ G_1, ..., G_n \Rightarrow
A$ as a shorthand for $ G_1 \Rightarrow  ... \Rightarrow
G_n \Rightarrow A$. Note that Horn formulas are special case of formulas. 

\begin{definition}
  Weak Head reduction context $ \mathcal{C} \  ::=  \ \bullet \ | \ \mathcal{C}\ e$
\end{definition}

\begin{definition}[Weak Head Reduction]
  \label{label}
  \fbox{$ e \leadsto e'$}
\[\begin{array}{c}

\infer{{\mathcal{C}[\mu \alpha . e]} \leadsto_\mu \mathcal{C}[[\mu \alpha . e/\alpha]e]}{} 
\quad \quad
\infer{\mathcal{C}[(\lambda \alpha . e) \ e'] \leadsto_\beta \mathcal{C}[[e'/\alpha] e] }{}

  \end{array}
\]

\end{definition}

Note that weak head reduction context do not allow reduction under the constant $\kappa$. It is 
more restricted than the context in Section \ref{preliminary}.

\begin{definition}[General Corecursive Resolution]

\[\begin{array}{c}
  \infer[\text{if}~(e : G_1,..., G_m \Rightarrow A) \in \Phi]
    {\Phi \vdash \sigma A \Downarrow \kappa~e_1 \cdots e_n}
    { \Phi \vdash \sigma G_1 \Downarrow e_1 \quad \cdots \quad \Phi \vdash \sigma G_n \Downarrow e_n 
      } 
\\ \\ 
    \infer
    {\Phi \vdash F \Downarrow \mu \alpha . e }{\Phi, \alpha : F \vdash F \Downarrow e & \mathrm{HNF}(e)} 
\quad\quad
    \infer
    {\Phi \vdash \underline{G} \Rightarrow B \Downarrow \lambda \underline{\alpha} . e}
    {\Phi, \underline{\alpha} : \underline{G} \vdash  B \Downarrow e}
  \end{array}
\]
\end{definition}


\begin{definition}[Howard's Type System]
\label{proofsystem}
\[
\begin{array}{c}
\infer[(\textsc{Assump})]{\Phi \vdash a : F}{(a :  F) \in \Phi}    
\quad\quad
\infer[(\textsc{App})]{\Phi \vdash e_2\ e_1 : F}{\Phi \vdash e_1 : F' & \Phi \vdash e_2 : F' \Rightarrow F}
\\
\\

\infer[(\textsc{Abs})]{\Phi \vdash \lambda \alpha.e : F' \Rightarrow F}{\Phi, \alpha : F' \vdash e : F}
\quad\quad
\infer[(\textsc{Gen})]{\Phi \vdash e: \forall x . F}{\Phi \vdash e :  F & x \notin \mathrm{FV}(\Phi)}

\\
\\
\infer[(\textsc{Inst})]{\Phi \vdash e : [t/x]F}{\Phi \vdash e : \forall x . F}
\quad\quad
\infer[(\textsc{Mu})]{\Phi \vdash \mu \alpha . e : F}{\Phi, \alpha : F \vdash e : F & \mathrm{HNF}(e)}
  \end{array}  
\]
\end{definition}

\begin{theorem}
  If $\Phi \vdash F \Downarrow e$, then $\Phi \vdash e : F$. 
\end{theorem}
\begin{proof}
  By induction on the derivation of $\Phi \vdash F \Downarrow e$. 
\end{proof}

\begin{theorem}
  If $\Phi \vdash e : F$, then $e$ is terminating with respect to weak head reduction. 
\end{theorem}
\begin{proof} Sketch. The proof is very similar to standard normalization proof for simply typed
lambda calculus. The only tricky rule is the \textsc{Mu} rule: 

\

\infer{\Phi \vdash \mu \alpha. e : \underline{G} \Rightarrow A}{\Phi, \alpha :  \underline{G} \Rightarrow A \vdash e : \underline{G} \Rightarrow A & \mathrm{HNF}(e)}

\

\noindent We want to show $\mu \alpha. e$ is in the reducible set of type $\underline{G} \Rightarrow A$. Since $e = \lambda \underline{\alpha}. \kappa\ \underline{e}$, we just need to show for any reducible $\underline{u}$ of type $\underline{G}$, we have $(\mu \alpha. e) \ \underline{u}\leadsto_\mu (\lambda \underline{\alpha}. \kappa\ [\mu \alpha. e/\alpha]\underline{e})\ \underline{u}$ is terminating. This is the case due to the expression is
guarded by $\kappa$. 
\end{proof}

\section{Examples}
\label{examples}

In this section, we show several examples with the prototype that we developed. They are
also available from the \texttt{examples} directory in \url{https://github.com/Fermat/corecursive-type-class}.

\subsection{Example 1}
\begin{verbatim}
module bush where
axiom (Eq a, Eq (f (f a))) => Eq (HBush f a)
axiom Eq (f (Mu f) a) => Eq (Mu f a)
axiom Eq Unit
auto Eq (Mu HBush Unit)
\end{verbatim}

The keyword \texttt{axiom} introduces an axiom and the keyword \texttt{auto}
requests the system to prove the formula automatically using the heuristic
corecursive hypothesis generation that we described in Section \ref{inv}. If we
save the above code in the \texttt{bush.asl} file, and, at the command line, type
\texttt{asl bush.asl}, then we get the following output:

\begin{verbatim}
Parsing success! 
Type Checking success! 
Program Definitions
  Ax0 :: (Eq (f (Mu f) a)) => Eq (Mu f a)
  = Ax0 
  Ax1 :: (Eq a, Eq (f (f a))) => Eq (HBush f a)
  = Ax1 
  Ax2 :: Eq Unit
  = Ax2 
  genLemm4 :: (Eq var_1) => Eq (Mu HBush var_1)
  = \ b0 . Ax0 (Ax1 b0 (genLemm4 (genLemm4 b0))) 
  goalLem3 :: Eq (Mu HBush Unit)
  = genLemm4 Ax2 
Axioms
  Ax2 :: Eq Unit
  Ax1 :: (Eq a, Eq (f (f a))) => Eq (HBush f a)
  Ax0 :: (Eq (f (Mu f) a)) => Eq (Mu f a)
Lemmas
  goalLem3 :: Eq (Mu HBush Unit)
  genLemm4 :: (Eq var_1) => Eq (Mu HBush var_1)
\end{verbatim}

The corecursive hypothesis generated is  \texttt{genLemm4 :: (Eq var\string_1) => Eq (Mu HBush var\string_1)}, its proof is \texttt{\string\ b0 . Ax0 (Ax1 b0 (genLemm4 (genLemm4 b0)))}. The proof
for \texttt{Eq (Mu HBush Unit)} is \texttt{genLemm4 Ax2}. 

Using \texttt{axiom} and \texttt{auto} allows us to quickly experiment with
different small examples. Here is the corresponding type-class code. 
\begin{verbatim}
module bush where
data Unit where
  Unit :: Unit

data Maybe a where
  Nothing :: Maybe a
  Just :: a -> Maybe a

data Pair a b where
  Pair :: a -> b -> Pair a b

data HBush f a where
  HBLeaf :: HBush f a
  HBNode ::  a -> (f (f a)) -> HBush f a

data Mu f a where
  In :: f (Mu f) a -> Mu f a
  
data Bool where
     True :: Bool
     False :: Bool

class Eq a where
   eq :: Eq a => a -> a -> Bool

and = \ x y . case x of
                True -> y
                False -> False

instance  => Eq Unit where
   eq = \ x y . case x of
                   Unit -> case y of 
                              Unit -> True

instance Eq a, Eq (f (f a)) => Eq (HBush f a) where
  eq = \ x y . case x of
                 HBLeaf -> case y of
                            HBLeaf -> True
                            HBNode a c -> False
                 HBNode a c1 -> case y of
                              HBLeaf -> False
                              HBNode b c2  -> and (eq a b) (eq c1 c2)

instance Eq (f (Mu f) a) => Eq (Mu f a) where
   eq = \ x y . case x of
                  In s -> case y of
 		            In t -> eq s t

term1 = In HBLeaf
term2 = In (HBNode Unit term1)
test = eq term2 term1
test1 = eq term2 term2
reduce test
reduce test1
\end{verbatim}

Inspecting the output of this code, we see that the result of the evaluation of 
\texttt{test} (resp. \texttt{test1}) is \texttt{False} (resp. \texttt{True}). It is quite
verbose and probably irrelevant to see the type-class code, so in the following
we will show examples using only the \texttt{axiom}, \texttt{auto} and \texttt{lemma} keywords.

\subsection{Example 2}

\begin{verbatim}
module lam where
axiom Eq (f (Mu f) a) => Eq (Mu f a)
axiom (Eq a, Eq (f a), Eq (f a), Eq (f (Maybe a))) => Eq (HLam f a)
axiom Eq Unit
axiom Eq a => Eq (Maybe a)
lemma (Eq x) => Eq (Mu HLam x)
lemma Eq (Mu HLam Unit)
\end{verbatim}

Of course, our heuristic \texttt{auto Eq (Mu HLam Unit)} also works for this case. But
we can specify the corecursive hypothesis as lemma and guided our mini-prover to prove the 
final goal \texttt{Eq (Mu HLam Unit)}.

\subsection{Example 3}
Are there any examples where \texttt{auto} fails, but where we can come to the rescue with a \texttt{lemma}? 
The answer is yes. 

\begin{verbatim}
axiom (Eq a, Eq (Pair (f1 a) (f2 a))) => Eq (H1 f1 f2 a)
axiom Eq (Pair (g a) (f (g a))) => Eq (H2 f g a)
axiom Eq (h1 (Mu h1 h2) (Mu h2 h1) a) => Eq (Mu h1 h2 a)
axiom (Eq a, Eq b) => Eq (Pair a b)
axiom Eq Unit
auto Eq (Mu H1 H2 Unit)
\end{verbatim}
If we run this code, our mini-prover diverges, since this example require multiple lemmas
in order to prove the final goal \texttt{Eq (Mu H1 H2 Unit)}. 

\begin{verbatim}
axiom (Eq a, Eq (Pair (f1 a) (f2 a))) => Eq (H1 f1 f2 a)
axiom Eq (Pair (g a) (f (g a))) => Eq (H2 f g a)
axiom Eq (h1 (Mu h1 h2) (Mu h2 h1) a) => Eq (Mu h1 h2 a)
axiom (Eq a, Eq b) => Eq (Pair a b)
axiom Eq Unit
lemma (Eq x, Eq (Mu H2 H1 x)) => Eq (Mu H1 H2 x)
lemma Eq x => Eq (Mu H2 H1 x)
lemma Eq (Mu H1 H2 Unit)
\end{verbatim}

\subsection{Example 4}
\label{DZ}
Are there any examples where even \texttt{lemma} does not work? We believe that
the following is such an example.
\begin{verbatim}
axiom D n (S m) => D (S n) m
axiom D (S m) Z => D Z m
\end{verbatim}

Note that \texttt{auto D Z Z} will not work because the corecursive hypothesis
generated by our method is not provable. The following is a solution in Haskell
using type families. We want to point out that using type families is a way to
introduce equality axioms for addition, and these equality axioms are not
derivable from the original axioms. The ability to learn addition seems to
require a higher notion of intelligence.    

\begin{verbatim}
{-# LANGUAGE Rank2Types, TypeFamilies, UndecidableInstances #-}
data Z
data S n
data D a b
type family Add m n
type instance Add n Z  =  n
type instance Add m (S n) = Add (S m) n

k1 :: D n (S m) -> D (S n) m
k1 = undefined
k2 :: D (S m) Z -> D Z m
k2 = undefined

f :: (forall n. D n (S m) -> D (S (Add m n)) Z) -> D Z m
f g = k2 (g (f (g . k1)))
\end{verbatim}

\end{document}